\documentclass[journal,twoside,web]{ieeecolor}
\usepackage{generic}
\usepackage{cite}
\usepackage{amsmath,amssymb,amsfonts}
\usepackage{multirow}
\usepackage{algorithmic}
\usepackage{graphicx}
\usepackage{diagbox}
\usepackage{algorithm,algorithmic}
\usepackage{hyperref}
\hypersetup{hidelinks=true}
\usepackage{arydshln}
\usepackage{textcomp}
\usepackage{tabu}
\usepackage{rotating}
\hypersetup{hidelinks=true}
\newtheorem{definition}{Definition}
\newtheorem{proposition}{Proposition}
\newtheorem{remark}{Remark}
\newtheorem{lemma}{Lemma}
\newtheorem{corollary}{Corollary}
\newtheorem{theorem}{\textbf{Theorem}}

\def\BibTeX{{\rm B\kern-.05em{\sc i\kern-.025em b}\kern-.08em
    T\kern-.1667em\lower.7ex\hbox{E}\kern-.125emX}}
\begin{document}
\title{Privacy-Aware State Estimation: From Coarse to Precise Privacy Protection}
\author{Zhongyao Hu, Jason J. R. Liu, \IEEEmembership{Senior Member, IEEE}, Jun Shang, \IEEEmembership{Senior~Member, IEEE},\\Zhan Shu, \IEEEmembership{Senior Member, IEEE}
\thanks{Manuscript received xxx; accepted xxx. Date of publication xx; date of current version xx.}
\thanks{This paper was recommended by Associate Editor xxx.}
\thanks{This work is supported in part by the National Natural Science Funds of China under Grant xxx, and in part by the xxx under Grant xxx. (Corresponding author: Jason J. R. Liu.)}% <-this % stops a space
\thanks{Zhongyao Hu is with the College of Information Engineering, Zhejiang University of Technology, Hangzhou 310023, China. (email: zyhu98@zjut.edu.cn).}
\thanks{Jason J. R. Liu is with the Department of Electromechanical Engineering, University of Macau, Macau, China (email: jasonliu@um.edu.mo).}
\thanks{Jun Shang is with the Department of Control Science and Engineering, Shanghai Institute of Intelligent Science and Technology, State Key Laboratory of Autonomous Intelligent Unmanned Systems, and Frontiers Science Center for Intelligent Autonomous Systems, Tongji University, Shanghai 200092, China (e-mail: shangjun@tongji.edu.cn).}
\thanks{Zhan Shu is with the Department of Electrical and Computer Engineering, University of Alberta, Edmonton, Alberta, T6G 1H9, Canada (email: zshu1@ualberta.ca).}}
\maketitle

\begin{abstract}
This paper addresses the problem of achieving both coarse and precise privacy in state estimation. Coarse privacy forces the eavesdropper's total mean-square error (MSE) to infinity, but errors along certain confidential directions may remain bounded. This motivates precise privacy, which additionally drives the MSE along prescribed directions to infinity. For coarse privacy, an analytical transformation is established, preserving the user's optimality and driving the eavesdropper's total MSE to infinity at a polynomial-exponential rate. A stochastic intermittent encryption scheme is further developed, and an explicit lower bound on the encryption probability is derived to guarantee divergence. For precise privacy, by analyzing the behavior of the Riccati equation on the unobservable subspace, we prove that the eavesdropper's directional MSE becomes unbounded if and only if the direction's unstable component lies outside the observable subspace. Finally, a systematic method is proposed to exclude target vectors from the observable subspace, forcing the directional MSE to infinity.
%This paper addresses privacy-aware state estimation with dual objectives: coarse privacy, forcing the eavesdropper's total mean-square error (MSE) to infinity, and precise privacy, driving the MSE along specific directions to infinity. A full-row-rank linear transformation is employed to compress and reorganize critical measurement information, enabling selective encryption of partial dimensions for privacy preservation. For coarse privacy, an analytic linear transformation constructed via matrix decomposition maintains the user's estimation optimality while compromising the eavesdropper's detectability. Monotonicity analysis proves the eavesdropper's total MSE grows unbound under this undetectability condition. For precise privacy, asymptotic analysis based on optimal filtering theory establishes that the eavesdropper's directional MSE becomes unbounded if and only if the direction's unstable component lies outside the observable subspace. This work establishes a theoretical foundation for hierarchical privacy preservation in state estimation with practical efficiency.
\end{abstract}


\section{Introduction}
Cyber-Physical Systems (CPSs), the backbone of Industry 4.0, depend on accurate state estimation for critical infrastructure operation. However, the open network environment exposes CPSs' data to illegal eavesdropping during information exchange \cite{Hassan8854247}. Once sufficient information is obtained, eavesdroppers can infer confidential information and design highly threatening cyberattacks to disrupt CPSs' functionality \cite{Yan10648957}. Consequently, designing privacy-aware state estimation methods to prevent eavesdroppers from acquiring sensitive system information has become an urgent priority.

Contaminating messages with artificial noise is one of the most common privacy protection strategies \cite{Le6606817,Degue9993779,Wang9678137,Leong8550317}. While this method reduces the signal-to-noise ratio of intercepted messages, it also degrades performance for legitimate users. Although the authors in \cite{Leong8550317} claimed that the user can receive intact messages if the added artificial noise lies in the null space of the channel, practical channel randomness often invalidates this condition.  Additionally, the artificial packet dropping strategy reduces the eavesdropper's information by discarding data packets \cite{Du8899554,Zhao9601251,TSIAMIS20178385,Leong8543618}. Random dropping can force the eavesdropper's mean-square error (MSE) to infinity, but it typically requires a better legitimate channel~\cite{Du8899554,Zhao9601251,TSIAMIS20178385}. A packet scheduling strategy was designed in \cite{Leong8543618} to balance performance between the user and the eavesdropper, but introducing additional communication overhead. 
%Similar to the artificial noise strategy, the packet dropping strategy also degrades users' estimation performance.
Another common privacy-preserving strategy is dynamic encoding, which incorporates system dynamics into data encoding \cite{Tsiamis8758381,Kennedy10491308,Lücke9762536,Marelli10684094,CRIMSON2025111932,YANG2020109116,AN2022110087}. A state-secrecy encoding scheme was developed in \cite{Tsiamis8758381,Kennedy10491308,Lücke9762536,Marelli10684094}, injecting historical data into current data to drive eavesdroppers' MSE toward open-loop levels. However, these methods require feedback communication, increasing communication burden \cite{Tsiamis8758381,Kennedy10491308,Lücke9762536,Marelli10684094}. The authors of \cite{CRIMSON2025111932} and \cite{YANG2020109116} introduced deceptive encoding to mislead the eavesdropper, but statistical inconsistencies in encoded data may allow inference through statistical testing. Moreover, two encoding methods were designed in \cite{AN2022110087} for distributed estimation to conceal the system's state from eavesdroppers without compromising estimation accuracy. Unlike the above strategies, cryptographic methods prevent privacy leakage by encrypting messages with keys while allowing legitimate users to reconstruct original data \cite{Shim7172449}. However, they require tedious manipulation to prevent ciphertext deciphering \cite{MISHRA2024101037}. To reduce encryption consumption, \cite{WANG2022110145,HUANG2021109537,Tao10621055} designed encryption scheduling algorithms, proving that optimal scheduling compacts encryption moments. However, encrypting the entire message creates instantaneous burdens. The authors of \cite{Shang9882330} and \cite{Zou10782997} proposed a partial encryption scheme to alleviate this instantaneous burden by encrypting only a part of measurements per moment. 
%Following this idea, the authors of \cite{Zou10782997} designed partial encryption for systems with packet loss. 

Existing encryption-based methods predominantly aim to maximize the eavesdropper's total MSE. However, the sensitivity of state variables is not uniform; it is intrinsically linked to specific directions that correspond to critical physical or operational variables (e.g., displacement along critical axes). A critical limitation persists as even with unbounded total MSE, the eavesdropper's directional MSE may remain finite, enabling confidential information inference. Actually, this limitation stems from their coarse design, which maximizes total MSE without distinguishing critical from non-critical directions. To bridge this gap, we employ a full-row-rank linear transformation to compress measurements such that critical information concentrates in leading dimensions, enabling directional encryption. 
The key contributions are:
%To bridge this gap, this paper proposes a directional encryption method through measurement reorganization. 
\begin{itemize}
    \item Through invariance analysis of the Riccati equation, we derive an analytical linear transformation (Corollary~1) that not only achieves lossless measurement compression preserving the user's optimality (Theorem~1), but also destroys the eavesdropper's detectability through directional encryption (Proposition~1).
    
    \item By exploiting the monotonicity of the Riccati equation, we prove that the eavesdropper's total MSE grows unbounded (Theorem 2), and we characterize its growth rate in terms of the encrypted unstable modes (Proposition 2). Leveraging this rate characterization, an intermittent scheme is designed to reduce encryption frequency, together with a lower bound on the encryption probability that still guarantees unboundedness (Proposition 3).
    
    \item We establish that the eavesdropper's MSE in a specific direction becomes unbounded if and only if the unstable component of that direction lies outside the eavesdropper's observable subspace (Propositions 4 and 5, and Theorem 3). A method is developed to exclude arbitrary vectors from the observable subspace, enabling the analytical design of the linear transformation (Corollary 2).
    %\item The orders of generalized eigenvectors are shown to represent information depth within system dynamics (e.g., velocity derivable from displacement implies deeper information depth for velocity). Rigorous analysis reveals a fundamental property: deeper information requires encrypting more measurement dimensions.
\end{itemize}

\textbf{Notations:} Let $\mathbb{R}^{n\times m}$ denote the space of $n\times m$ real matrices, and $\mathbb{R}^n$ the $n$-dimensional Euclidean space. The $n\times n$ identity matrix and the $n\times m$ zero matrix are denoted by $I_n$ and $0_{n\times m}$ (or $0_n$ if $n=m$), respectively; the subscripts are omitted when the dimensions are clear from the context. The canonical basis vector $e_{n,i} \in \mathbb{R}^{n}$ is a vector with $1$ at the $i$-th entry and $0$ elsewhere. The direct sum and the Kronecker product of matrices are denoted by $\oplus$ and $\otimes$, respectively. For a matrix, let $\mathrm{Tr}(\cdot)$ denote the trace, $\|\cdot\|$ the operator norm, $\mathrm{R}(\cdot)$ the range, $\mathrm{N}(\cdot)$ the null space, and $\rho(\cdot)$ the spectral radius. For symmetric matrices $X, Y \in \mathbb{R}^{n\times n}$, $X>Y$ (resp. $X\geq Y$) means that $X-Y$ is positive definite (resp. positive semi-definite). $\lambda_{\min}(\cdot)$ denotes the minimum eigenvalue of a symmetric matrix. Horizontal concatenation of matrices is written as $[\cdots, \cdots, \cdots]$ and vertical concatenation as $[\cdots; \cdots; \cdots]$. For an index set $\mathbb{S}=\{s_1,s_2,\cdots,s_r\}$ with $1\leq s_1 < s_2 < \cdots < s_r \leq m$ and a matrix $X = [x_1, x_2, \cdots, x_m] \in \mathbb{R}^{n\times m}$, the column selection operation is defined as $[X]_{\mathbb{S}} \triangleq [x_{s_1}, x_{s_2}, \cdots, x_{s_r}]$. The singular value decomposition (SVD) of $X$ is written as $X \stackrel{\mathrm{svd}}{=} U\Sigma V^{\mathrm T}$. The cardinality of a finite set $\mathbb{S}$ is denoted by $|\mathbb{S}|$, and the mathematical expectation by $\mathrm{E}[\cdot]$. $\delta_{ij}$ denotes the Kronecker delta, i.e., $\delta_{ij}=1$ if $i=j$ and $\delta_{ij}=0$ otherwise.

For matrices $A$, $B$, $C$, $Q>0$, $S$, $R>0$, and $X\geq0$ with compatible dimensions, we define the operators 
\begin{equation*}\begin{aligned}
\mathfrak{R}(X;A,B,C,Q,S,R)&\triangleq AXA^\mathrm{T}+BQB^\mathrm{T}\\
-(AXC^\mathrm{T}+BS)&(CXC^\mathrm{T}+R)^{-1}(AXC^\mathrm{T}+BS)^\mathrm{T},\\
\end{aligned}\end{equation*}
\begin{equation*}\begin{aligned}
\mathfrak{K}(X;A,B,C,S,R)&\triangleq(AXC^\mathrm{T}+BS)(CXC^\mathrm{T}+R)^{-1},\\
\mathfrak{L}(X;A,B,Q)&\triangleq AXA^\mathrm{T}+BQB^\mathrm{T}.
\end{aligned}\end{equation*}
The $k$-fold composition of the operators $\mathfrak{R}$ and $\mathfrak{L}$ with respect to $X$ are defined as $\mathfrak{L}_k(X;A,B,Q)$ and $\mathfrak{R}_k(X;A,B,C,Q,S,R)$, respectively. By convention, $\mathfrak{L}_0(X;A,B,Q)\triangleq X$ and $\mathfrak{R}_0(X;A,B,C,Q,S,R)\triangleq X$. If $C=0$, $S=0$, and $R=0$, we define  
\begin{equation*}\begin{aligned}
\mathfrak{R}(X;A,B,0,Q,0,0)&\triangleq \mathfrak{L}(X;A,B,Q),\\
\mathfrak{K}(X;A,B,0,0,0)&\triangleq 0.
\end{aligned}\end{equation*}

Moreover, the controllability matrix, observability matrix, and Topeliz matrix are defined as 
\begin{equation*}\begin{aligned}
\mathfrak{C}_n(A,B)&\triangleq 
\left[A^{n-1}B,A^{n-2}B,\cdots,B\right],\\
\mathfrak{O}_n(C,A)&\triangleq \left[
C;CA;\cdots;CA^{n-1}
\right],\\
\mathfrak{T}_n(A,B,C,D)&\triangleq
\begin{bmatrix}
D&\cdots&\cdots&\cdots&0\\
CB&\ddots&0&0&\vdots\\
CAB&CB&\ddots&0&\vdots\\
\vdots&\ddots&\ddots&\ddots&\vdots\\
CA^{n-2}B&\cdots &CAB&CB&D
\end{bmatrix}.
\end{aligned}\end{equation*}

\section{Problem Formulation}
\subsection{System Description}
Consider the discrete-time linear system
\begin{equation}
\left\{ \begin{array}{l}
x(k+1)=Ax(k)+Bw(k),\\
z(k)=Cx(k)+v(k),
\end{array} \right.
\end{equation}
where $A\in\mathbb{R}^{n\times n}$, $B\in\mathbb{R}^{n\times l}$, $C\in\mathbb{R}^{m\times n}$, $x(k)\in\mathbb{R}^n$ and $z(k)\in\mathbb{R}^m$ are respectively the state and measurement, $\{w(k)\}_{k\geq0}$ and $\{v(k)\}_{k\geq0}$ are Gaussian white noise sequences. The initial state $x(0)$ is Gaussian with mean $0$, $\rho(A)\geq1$, $(A,B)$ is stabilizable and $(C,A)$ is detectable. The second moments of the initial state and noises are given by 
\begin{equation*}\begin{aligned}
\mathrm{E}\left[\begin{bmatrix}
{x(0)}\\
{w(i)}\\
{v(i)}
\end{bmatrix}\begin{bmatrix}
{x(0)}\\
{w(j)}\\
{v(j)}
\end{bmatrix}^{\mathrm{T}}\right]=
\begin{bmatrix}
{P(0)} & {0} & {0}\\
{0} & {\delta_{ij}Q} & {\delta_{ij}S}\\
{0} & {\delta_{ij}S^{\mathrm{T}}} & {\delta_{ij}R}
\end{bmatrix},
\end{aligned}\end{equation*}
where $P(0)\geq0$ and $\bigl[\begin{smallmatrix} Q & S \\ S^{\mathrm{T}} & R \end{smallmatrix}\bigr] > 0$.

According to linear system theory \cite{chen1984linear}, one can perform a controllable decomposition on the system (1) to obtain 
\begin{equation*}\begin{aligned}
&A_1=T_1^{-1}AT_1=
\begin{bmatrix}
{A_\mathrm{c}} & {A_{\prime}}\\
{0} & {A_{\bar{\mathrm{c}}}}
\end{bmatrix},\\ 
&B_1=T_1^{-1}B=\begin{bmatrix}
{B_\mathrm{c}}\\
{0}
\end{bmatrix},\ C_1=CT_1,
\end{aligned}\end{equation*}
where $A_\mathrm{c}\in\mathbb{R}^{n_\mathrm{c}\times n_\mathrm{c}}$, $A_{\bar{\mathrm{c}}}\in\mathbb{R}^{n_{\bar{\mathrm{c}}}\times n_{\bar{\mathrm{c}}}}$, $B_\mathrm{c}\in\mathbb{R}^{n_\mathrm{c}\times l}$, $n=n_\mathrm{c}+n_{\bar{\mathrm{c}}}$, $n_\mathrm{c}=\mathrm{rank}(\mathfrak{C}_n(A,B))$, $(A_\mathrm{c},B_\mathrm{c})$ is reachable, and $\rho(A_{\bar{\mathrm{c}}})<1$. 

Denote the real Jordan canonical form \cite{horn2012matrix} of $A_\mathrm{c}$ as $T_2^{-1}A_\mathrm{c}T_2=J=J_1\oplus J_2\oplus\cdots\oplus J_r$, where each real Jordan block $J_i$ is structured as 
\begin{equation}
J_i=
\begin{bmatrix}
\Lambda_i   &I          &0      &0          \\
0     &\Lambda_i        &\ddots &0          \\
0     &0          &\ddots &I          \\
0     &0          &0      &\Lambda_i
\end{bmatrix}\in\mathbb{R}^{d_i\times d_i},
\end{equation}
and $r$ and $d_i$ are constants determined by Weyr characteristics of $A_\mathrm{c}$. In the real Jordan canonical form, $\Lambda_i=\lambda_i$ if $J_i$ associates with a real eigenvalue $\lambda_i$, and $\Lambda_i=\left[\begin{smallmatrix}
\operatorname{Re}(\lambda_i) & \operatorname{Im}(\lambda_i) \\
-\operatorname{Im}(\lambda_i) & \operatorname{Re}(\lambda_i)
\end{smallmatrix}\right]$ if $J_i$ associates with a pair of conjugate complex eigenvalues $\lambda_i$ and $\bar{\lambda}_i$\footnote{For a complex matrix $X$, $\bar{X}$, $\mathrm{Re}(X)$, and $\mathrm{Im}(X)$ denote its conjugate matrix, real part, and imaginary part, respectively}. Since the real Jordan blocks can appear in any order utilizing a block permutation similarity, it is assumed that $\rho(J_1)\geq\cdots\geq\rho(J_{r_\mathrm{u}})\geq1>\rho(J_{r_\mathrm{u}+1})\cdots\geq\rho(J_r)$, where $r_\mathrm{u}$ is the number of unstable real Jordan blocks. 

Define $A_\mathrm{u}\triangleq J_1\oplus\cdots\oplus J_{r_\mathrm{u}}$ and $A_\mathrm{s}\triangleq J_{r_\mathrm{u}+1}\oplus\cdots\oplus J_r$, one has  
\begin{equation}\begin{aligned}
&A_2=(T_2^{-1}\oplus I)A_1(T_2\oplus I)=
\begin{bmatrix}
{A_\mathrm{u}} & {0} & {A_{\prime1}}\\
{0} & {A_\mathrm{s}} & {A_{\prime2}}\\
{0} & {0} & {A_{\bar{\mathrm{c}}}}
\end{bmatrix},\\
&B_2=(T_2^{-1}\oplus I)B=\begin{bmatrix}
{T_2^{-1}B_\mathrm{c}}\\
{0}
\end{bmatrix}=\begin{bmatrix}
{B_\mathrm{u}}\\
{B_\mathrm{s}}\\
{0}
\end{bmatrix},\\
&C_2=C_1(T_2\oplus I)=\begin{bmatrix}
{C_\mathrm{u}} & {C_\mathrm{s}} & {C_{\bar{\mathrm{c}}}}
\end{bmatrix},
\end{aligned}\end{equation}
where $A_\mathrm{u}\in\mathbb{R}^{n_\mathrm{u}\times n_\mathrm{u}}$, $A_\mathrm{s}\in\mathbb{R}^{n_\mathrm{s}\times n_\mathrm{s}}$, $B_\mathrm{u}\in\mathbb{R}^{n_\mathrm{u}\times l}$, $B_\mathrm{s}\in\mathbb{R}^{n_\mathrm{s}\times l}$,  $C_\mathrm{u}\in\mathbb{R}^{m\times n_\mathrm{u}}$, $C_\mathrm{s}\in\mathbb{R}^{m\times n_\mathrm{s}}$, $C_{\bar{\mathrm{c}}}\in\mathbb{R}^{m\times n_{\bar{\mathrm{c}}}}$, and $n_\mathrm{c}=n_\mathrm{u}+n_\mathrm{s}$. Since invertible linear transformations do not change the reachability, $\big([\begin{smallmatrix}
A_{\mathrm{u}} & 0\\
0 & A_{\mathrm{s}}
\end{smallmatrix}],
[\begin{smallmatrix}
B_{\mathrm{u}}\\
B_{\mathrm{s}}
\end{smallmatrix}]\big)$ is still reachable. 

Unless explicitly stated, it is assumed that the system (1) has been transformed into the form shown in (3), i.e., $A=A_2$, $B=B_2$, and $C=C_2$. 

%\begin{remark}
%By means of the controllable decomposition, the reachable Gramian on the controllable subspace is symmetric positive definite. Meanwhile, the unstable modes of the system can be explicitly represented through the real Jordan decomposition. These two properties enable us to derive explicit theoretical results.
%%which makes it easier to deflate Lyapunov and Riccati difference equations (Propositions 4 and 5, Theorem 4). Meanwhile, the unstable modes of the system can be explicitly represented through the real Jordan decomposition, thus showing more clearly the relationship between the measurement and unstable modes (Proposition 3, Theorems 3 and 4). These two decompositions will help us to analyze the effect of unstable modes of the system on the MMSE estimate.
%\end{remark}
\subsection{Problems of Interest}
To prevent privacy leakage, messages shall be encrypted before transmitting. However, due to the high computational burden of encryption algorithms, encrypting the whole message may be unaffordable. To solve the problem, a partial encryption scheme will be utilized. Specifically, the sensor encrypts only the first $m_\mathrm{c}$ components of $L_\mathrm{u}z(k)$, where $L_\mathrm{u}\in\mathbb{R}^{(m_\mathrm{c}+m_\mathrm{e})\times m}$ has full row rank and $m_\mathrm{c}+m_\mathrm{e}\leq m$. Partition $L_\mathrm{u}=[L_\mathrm{c};L_\mathrm{e}]$, where $L_\mathrm{c}\in\mathbb{R}^{m_\mathrm{c}\times m}$ and $L_\mathrm{e}\in\mathbb{R}^{m_\mathrm{e}\times m}$. Then, the encryption can be expressed by
\begin{equation}\begin{aligned}
\mathfrak{E}_\kappa(L_\mathrm{c}z(k)),
\end{aligned}\end{equation}
where $\mathfrak{E}_\kappa(\cdot):\mathbb{R}^{m_\mathrm{c}}\to \mathbb{R}^{m_\mathrm{c}}$ is the encryption function and $\kappa$ is the key. Based on the discussion above, the transmitted message can be represented by $[\mathfrak{E}_\kappa(L_\mathrm{c}z(k));L_\mathrm{e}z(k)]$.

In this paper, all system parameters $A$, $B$, $C$, $Q$, $S$, $R$, $P(0)$, $L_\mathrm{u}$, $m_\mathrm{c}$, and $m_\mathrm{e}$ are public knowledge, known to both the user and the eavesdropper. The only parameter not available to the eavesdropper is the key $\kappa$. Moreover, the eavesdropper is assumed to be capable of intercepting the message over the sensor-to-user channel. With the key $\kappa$, the user can decrypt the ciphertext and recover the original message $L_\mathrm{u}z(k)$. In contrast, the eavesdropper, lacking $\kappa$, cannot decrypt the ciphertext portion of the intercepted message. Thus, the eavesdropper can only access the plaintext component $L_\mathrm{e}z(k)$.

Based on the discussion above, the minimum MSE (MMSE) estimate of the user and the eavesdropper can be expressed as
\begin{equation}\begin{aligned}
\hat{x}(L_\iota,k+1)&\triangleq\mathrm{E}[x(k+1)|\mathbb{Z}(L_\iota,k)],\\
\hat{P}(L_\iota,k+1)&\triangleq\mathrm{E}[\tilde{x}(L_\iota,k+1)\tilde{x}(L_\iota,k+1)^\mathrm{T}|\mathbb{Z}(L_\iota,k)],
\end{aligned}\end{equation}
where $\iota\in\{\mathrm{u},\mathrm{e}\}$, $\tilde{x}(L,k)\triangleq x(k)-\hat{x}(L,k)$, and $\mathbb{Z}(L,k)\triangleq\{Lz(0),\cdots,Lz(k)\}$. 
Particularly, when $m_\mathrm{e}=0$, no usable information is available to the eavesdropper. For ease of presentation, we define $L_\mathrm{e}=0$, $S=0$, and $R=0$ if $m_\mathrm{e}=0$. 

This paper will design the parameters $m_\mathrm{c}$, $m_\mathrm{e}$, and $L_\mathrm{u}$ to achieve two different kinds of privacy, as shown below.

\begin{definition}
We say that \textbf{\em coarse privacy} is achieved if
\begin{itemize}
    \item the estimation performance of the user is preserved, i.e., $\hat{x}(L_\mathrm{u},k)=\hat{x}(I,k)$ and $\hat{P}(L_\mathrm{u},k)=\hat{P}(I,k)$ for $k\geq1$; 
    \item and the eavesdropper's total estimation capability is completely degraded, i.e., $\lim_{k\rightarrow\infty}\mathrm{E}\left[\|\tilde{x}(L_\mathrm{e},k)\|^2\right]=\infty$.
\end{itemize}
\end{definition}

\begin{definition}
For a given nonzero vector $\varphi\in\mathbb{R}^{n}$, we say that \textbf{\em $\varphi$-precise privacy} is achieved if 
\begin{itemize}
    \item $\hat{x}(L_\mathrm{u},k)=\hat{x}(I,k)$ and $\hat{P}(L_\mathrm{u},k)=\hat{P}(I,k)$ for $k\geq1$; 
    \item and the eavesdropper's estimation capability on $\varphi^{\mathrm{T}}x(k)$ is fully destroyed, i.e., $\lim_{k\rightarrow\infty}\mathrm{E}\left[\|\varphi^{\mathrm{T}}\tilde{x}(L_\mathrm{e},k)\|^2\right]=\infty$.
\end{itemize}
\end{definition}

\begin{remark}
For system (1), confidential information often manifests as $\varphi^{\mathrm{T}}x_k$, where $\varphi$ can be interpreted as the confidential direction of the system (e.g., a single bus voltage among all grid states, or a vehicle coordinate within its full motion state). While Definition 1 ensures $\lim_{k \to \infty}\mathrm{E}\left[\|\tilde{x}(L_\mathrm{e}, k)\|^2\right]=\infty$, this may not preclude $\lim_{k\to\infty}\mathrm{E}\left[\|\varphi^{\mathrm{T}}\tilde{x}(L_\mathrm{e}, k)\|^2\right]<\infty$ for some $\varphi$. For example, $\lim_{k\to\infty}\|[k;-k]\|=\infty$ but $\lim_{k\to\infty}\|[1,1][k;-k]\|=0$. Thus, Definition 1 is not sufficient to protect specific confidential information. Definition~2 strengthens the privacy requirement by explicitly demanding the divergence of $\varphi^{\mathrm{T}}\tilde{x}(L_\mathrm{e},k)$, and it constitutes the first formal treatment of direction-specific privacy in the literature.
\end{remark}

%\begin{remark}
%The matrix $L_\mathrm{u}$ serves two key purposes: 1) it reduces the measurement dimension, lowering communication overhead; and 2) it concentrates all critical information within the first $m_\mathrm{c}$ components, allowing partial encryption.
%\end{remark}

\section{Coarse Privacy}
In this section, we discuss how to achieve coarse privacy. Based on Kalman filtering theory \cite{anderson2005optimal}, one can represent the MMSE estimators of the user and eavesdropper as 
\begin{equation}\begin{aligned}
&\hat{x}(L_\iota,k+1)\\
&\quad\quad\quad= A\hat{x}(L_\iota,k)+K(L_\iota,k)L_\iota\big(z(k)-C\hat{x}(L_\iota,k)\big),
\end{aligned}\end{equation}
\begin{equation}\begin{aligned}
&\hat{P}(L_\iota,k+1)\\
&\quad\quad\quad=\mathfrak{R}(\hat{P}(L_\iota,k);A,B,L_\iota C,Q,SL_\iota^{\mathrm{T}},L_\iota RL_\iota^{\mathrm{T}}),
\end{aligned}\end{equation}
\begin{equation}\begin{aligned}
K(L_\iota,k)=\mathfrak{K}(\hat{P}(L_\iota,k);A,B,L_\iota C,SL_\iota^{\mathrm{T}},L_\iota RL_\iota^{\mathrm{T}}),
\end{aligned}\end{equation}
where $\hat{x}(L_\iota,0)\triangleq0$, $\hat{P}(L_\iota,0)\triangleq P(0)$, and $\iota\in\{\mathrm{u},\mathrm{e}\}$. 

\subsection{Analytic parameters design}
By observing (6)--(8), one can find that in order not to affect the user, the linear transformation $L_\mathrm{u}$ should satisfy the following two matrix equations for every $P \geq 0$: 
\begin{equation}\begin{aligned}
&\mathfrak{R}(P;A,B,L_\mathrm{u}C,Q,SL_\mathrm{u}^\mathrm{T},L_\mathrm{u}RL_\mathrm{u}^\mathrm{T})\\
&\quad\quad\quad\quad\quad\quad\quad\quad=\mathfrak{R}(P;A,B,C,Q,S,R),
\end{aligned}\end{equation}
\begin{equation}\begin{aligned}
&\mathfrak{K}(P;A,B,LC,SL_\mathrm{u}^\mathrm{T},L_\mathrm{u}RL_\mathrm{u}^\mathrm{T})L_\mathrm{u}\\
&\quad\quad\quad\quad\quad\quad\quad\quad=\mathfrak{K}(P;A,B,C,S,R).
\end{aligned}\end{equation}
This can be compactly expressed as $L_\mathrm{u}\in\mathbb{L}$, where 
\begin{equation*}\begin{aligned}
\mathbb{L}\triangleq
\left\{X=L_\mathrm{u}:
\begin{array}{ll}
L_\mathrm{u}\ \text{has full row rank}\\
\text{(9) and (10) hold for every } P\geq0
\end{array}
\right\}.
\end{aligned}\end{equation*}
The following proposition will give an analytical subset of $\mathbb{L}$.

\begin{theorem}
\textit{Let $R^{-\frac{1}{2}}D\overset{\mathrm{svd}}{=}U\Sigma V^\mathrm{T}$, where $D\triangleq[C,(BS)^{\mathrm{T}}]$, $U\in\mathbb{R}^{m\times m}$, $\Sigma\in\mathbb{R}^{m\times 2n}$, and $V\in\mathbb{R}^{2n\times 2n}$. Then, $L_\mathrm{u}\in\mathbb{L}$ if $\mathrm{rank}(D)\leq m_{\mathrm{c}}+m_{\mathrm{e}}\leq m$ and 
\begin{equation}\begin{aligned}
L_\mathrm{u}=X\begin{bmatrix}
{I_{\operatorname{rank}(D)}} & {0}\\
{0} & {Y}
\end{bmatrix}(R^{\frac{1}{2}}U)^{-1},
\end{aligned}\end{equation}
where $X\in\mathbb{R}^{(m_\mathrm{c}+m_\mathrm{e})\times(m_\mathrm{c}+m_\mathrm{e})}$ is an arbitrary invertible matrix, and $Y\in\mathbb{R}^{(m_\mathrm{c}+m_\mathrm{e}-\mathrm{rank}(D))\times (m-\mathrm{rank}(D))}$ is an arbitrary full-row-rank matrix.}
\end{theorem}

\begin{proof}
From the definitions of $\mathfrak{R}(\cdot)$ and $\mathfrak{K}(\cdot)$, a direct algebraic expansion yields
\begin{equation*}\begin{aligned}
&\mathfrak{R}(P;A,B,LC,Q,SL^\mathrm{T},LRL^\mathrm{T})\\
=&APA^\mathrm{T}+BQB^\mathrm{T}\\
&-\mathfrak{K}(P;A,B,LC,SL^\mathrm{T},LRL^\mathrm{T})L(CPC^{\mathrm{T}}+R)\\
&\times \big(\mathfrak{K}(P;A,B,LC,SL^\mathrm{T},LRL^\mathrm{T})L\big)^\mathrm{T}.
\end{aligned}\end{equation*}
Thus, if (10) is satisfied, then (9) follows immediately. Consequently, $\mathbb{L}$ can be equivalently described as 
\begin{equation*}\begin{aligned}
\mathbb{L}=
\left\{X=L_\mathrm{u}:
\begin{array}{ll}
L_\mathrm{u}\ \text{has full row rank}\\
\text{(10) holds for all } P\geq0
\end{array}
\right\}.
\end{aligned}\end{equation*}

By applying the Woodbury matrix identity \cite{horn2012matrix}, one obtains
\begin{equation}\begin{aligned}
&\mathrel{\phantom{=}}\mathfrak{K}(P;A,B,L_\mathrm{u}C,SL_\mathrm{u}^\mathrm{T},L_\mathrm{u}RL_\mathrm{u}^\mathrm{T})L_\mathrm{u}\\
%&=(APC^\mathrm{T}+BS)L^{\mathrm{T}}(LCPC^\mathrm{T}L^{\mathrm{T}}+LRL^{\mathrm{T}})^{-1}L\\
%&=(APC^\mathrm{T}+BS)L^\mathrm{T}(LRL^\mathrm{T})^{-1}L\\
%&\quad-(APC^\mathrm{T}+BS)L^\mathrm{T}(LRL^\mathrm{T})^{-1}LCP^{\frac{1}{2}}\\
%&\quad\times\big(I+P^{\frac{1}{2}}C^\mathrm{T}L^\mathrm{T}(LRL^\mathrm{T})^{-1}LCP^{\frac{1}{2}}\big)^{-1}\\
%&\quad\times P^{\frac12}C^{T}L^\mathrm{T}(LRL^\mathrm{T})^{-1}L\\
&=(AP\Upsilon(C^{\mathrm{T}},L_\mathrm{u})+\Upsilon(BS,L_\mathrm{u}))\\
&\mathrel{\phantom{=}}-(AP\Upsilon(C^{\mathrm{T}},L_\mathrm{u})+\Upsilon(BS,L_\mathrm{u}))CP^{\frac{1}{2}}\\
&\mathrel{\phantom{=}}\times\big(I+P^{\frac{1}{2}}\Upsilon(C^{\mathrm{T}},L_\mathrm{u})CP^{\frac{1}{2}}\big)^{-1}P^{\frac12}\Upsilon(C^{\mathrm{T}},L_\mathrm{u}),
\end{aligned}\end{equation}
where $\Upsilon(X,Y)\triangleq XY(YRY^\mathrm{T})^{-1}Y$. Hence (10) holds for every $P\geq 0$ provided
\begin{equation}
\Upsilon(C^\mathrm{T},L_\mathrm{u}) = \Upsilon(C^\mathrm{T},I),\ 
\Upsilon(BS,L_\mathrm{u}) = \Upsilon(BS,I).
\end{equation}

Using the SVD $R^{-\frac{1}{2}}D\overset{\mathrm{svd}}{=}U\Sigma V^\mathrm{T}$, one verifies by direct substitution that the two equalities in (13) can be compactly reformulated as the projection equation
\begin{equation*}
(L_\mathrm{u}R^{\frac12}U)^\mathrm{T}\bigl( L_\mathrm{u}R^{\frac12}U (L_\mathrm{u}R^{\frac12}U)^\mathrm{T} \bigr)^{-1}L_\mathrm{u}R^{\frac12}U\Sigma = \Sigma.
\end{equation*}
By the property of orthogonal projection, this projection equation is equivalent to 
\begin{equation}
\mathrm{R}(\Sigma)\subseteq \mathrm{R}((L_\mathrm{u}R^{\frac{1}{2}}U)^\mathrm{T}).
\end{equation}

Since $\mathrm{rank}(\Sigma)=\mathrm{rank}(D)$ and $L_\mathrm{u}$ should have full row rank, there exists $L_\mathrm{u}$ such that (14) holds only if $\mathrm{rank}(D)\leq m_{\mathrm{c}}+m_{\mathrm{e}}\leq m$. 

Note that $\operatorname{Span}(e_{m,1},\cdots,e_{m,\operatorname{rank}(D)}) = \operatorname{R}(\Sigma)$, thus (14) holds if and only if $(L_\mathrm{u}R^{\frac{1}{2}}U)^\mathrm{T}$ can be formulated as 
\[
\begin{bmatrix} e_{m,1} & \cdots & e_{m,\operatorname{rank}(D)} & \psi_{\operatorname{rank}(D)+1} & \cdots & \psi_{m_{\mathrm{c}}+m_{\mathrm{e}}} \end{bmatrix} X,
\]
where $X\in\mathbb{R}^{(m_{\mathrm{c}}+m_{\mathrm{e}})\times(m_{\mathrm{c}}+m_{\mathrm{e}})}$ is an arbitrary invertible matrix and the vectors $\psi_{\operatorname{rank}(D)+1},\cdots,\psi_{m_{\mathrm{c}}+m_{\mathrm{e}}}\in\mathbb{R}^{m}$ are chosen so that 
$\{e_{m,1},\cdots,e_{m,\operatorname{rank}(D)},\psi_{\operatorname{rank}(D)+1},\cdots,\psi_{m_{\mathrm{c}}+m_{\mathrm{e}}}\}$ is linearly independent. Subtracting from each $\psi_i$ its orthogonal projection onto $\operatorname{span}(e_{m,1},\cdots,e_{m,\operatorname{rank}(D)})$ zeroes out the first $\operatorname{rank}(D)$ entries while preserving linear independence. 
Consequently, $\operatorname{R}(\Sigma)\subseteq \operatorname{R}((L_\mathrm{u}R^{\frac{1}{2}}U)^\mathrm{T})$ is equivalent to
\[
(L_\mathrm{u}R^{\frac{1}{2}}U)^\mathrm{T} = 
\begin{bmatrix}
I_{\operatorname{rank}(D)} & 0 \\
0 & Y
\end{bmatrix} X,
\]
where $Y\in\mathbb{R}^{(m-\operatorname{rank}(D))\times (m_{\mathrm{c}}+m_{\mathrm{e}}-\operatorname{rank}(D))}$ is an arbitrary matrix has rank $m_\mathrm{c}+m_\mathrm{e}-\mathrm{rank}(D)$.
\end{proof}

It should be emphasized that the analytical solution in Theorem 1 may not encompass all lossless linear transformations, since the proof relies on the sufficient condition (13) to solve (12). This restrictive approach is unavoidable, as analytically characterizing all solutions to the underlying nonlinear matrix equation (12) is inherently difficult.

The following theorem presents a necessary and sufficient condition for the eavesdropper's MSE to approach infinity.

\begin{theorem}
\textit{Consider the system (1). The condition $\lim_{k\to\infty}\mathrm{E}[\|\tilde{x}(L_\mathrm{e},k)\|^2]=\infty$ holds
if and only if $(L_\mathrm{e}C,A)$ is not detectable.}
\end{theorem}

\begin{proof}
Define
\begin{equation*}\begin{aligned}
A_\mathrm{e} &\triangleq \begin{cases}
A - BSL_\mathrm{e}^\mathrm{T}(L_\mathrm{e}RL_\mathrm{e}^\mathrm{T})^{-1}L_\mathrm{e}C, & \text{if } m_\mathrm{e} \neq 0, \\
A, & \text{if } m_\mathrm{e} = 0,
\end{cases}\\
Q_\mathrm{e} &\triangleq \begin{cases}
Q-SL_\mathrm{e}^\mathrm{T}(L_\mathrm{e}RL_\mathrm{e}^\mathrm{T})^{-1}L_\mathrm{e}S^\mathrm{T}, & \text{if } m_\mathrm{e} \neq 0, \\
Q, & \text{if } m_\mathrm{e} = 0.
\end{cases}
\end{aligned}\end{equation*}
By the Schur complement lemma, $Q_\mathrm{e}>0$. Then, the error covariance recursion can be rewritten as
\begin{equation}\begin{aligned}
&\hat{P}(L_\mathrm{e},k+1)\\
&\quad\quad\quad=\mathfrak{R}\big(\hat{P}(L_\mathrm{e},k);A_\mathrm{e},B,L_\mathrm{e}C,Q_\mathrm{e},0,L_\mathrm{e}RL_\mathrm{e}^\mathrm{T}\big).
\end{aligned}\end{equation}

If $(L_\mathrm{e}C,A)$ is detectable, then by the invariance of detectability under feedback~\cite{Anderson0319002}, $(L_\mathrm{e}C,A_\mathrm{e})$ is also detectable. From Kalman filtering theory, it follows that
\begin{equation*}\begin{aligned}
\lim_{k\to\infty}\mathrm{E}\left[\|\tilde{x}(L_\mathrm{e},k)\|^2\right]=\lim_{k\to\infty}\mathrm{Tr}\big(\hat{P}(L_\mathrm{e},k)\big)<\infty.
\end{aligned}\end{equation*}

If $(L_\mathrm{e}C,A)$ is not detectable, then $(L_\mathrm{e}C,A_\mathrm{e})$ is also not detectable. Consider the auxiliary recursion
\begin{equation*}\begin{aligned}
X(k)=\mathfrak{R}(X(k-1);A_\mathrm{e},B,L_\mathrm{e}C,Q_\mathrm{e},0,L_\mathrm{e}RL_\mathrm{e}^\mathrm{T}),
\end{aligned}\end{equation*}
with the initial condition $X(0)=0$. Clearly, $X(1)=BQ_\mathrm{e}B^\mathrm{T}\geq X(0)=0$. By the monotonicity of the Riccati equation~\cite{kailath2000linear}, we have
\begin{equation*}\begin{aligned}
X(2)\geq\mathfrak{R}(X(0);A_\mathrm{e},B,L_\mathrm{e}C,Q_\mathrm{e},0,L_\mathrm{e}RL_\mathrm{e}^\mathrm{T})=X(1).
\end{aligned}\end{equation*}
Proceeding inductively, $X(k)\geq X(k-1)$ for all $k\geq1$. Moreover, since $(A,B)$ is stabilizable, $(A_\mathrm{e},B)$ is also stabilizable by the invariance of stabilizability under feedback. Consequently, the algebraic Riccati equation
\begin{equation*}
X=\mathfrak{R}(X;A_\mathrm{e},B,L_\mathrm{e}C,Q_\mathrm{e},0,L_\mathrm{e}RL_\mathrm{e}^\mathrm{T})
\end{equation*}
admits no symmetric positive semi-definite solution~\cite{kailath2000linear}, which implies that the sequence $\{X(k)\}_{k\geq0}$ diverges to infinity.

By the monotonicity of the Riccati equation, $\hat{P}(L_\mathrm{e},k)\geq X(k)$ for all $k\geq0$. Hence,
\begin{equation*}\begin{aligned}
\lim_{k\to\infty}\mathrm{E}\left[\|\tilde{x}(L_\mathrm{e},k)\|^2\right]=\lim_{k\to\infty}\mathrm{Tr}\big(\hat{P}(L_\mathrm{e},k)\big)=\infty.
\end{aligned}\end{equation*}
This completes the proof.
\end{proof}

We define 
\begin{equation*}\begin{aligned}
    \check{d}_i&\triangleq
    \begin{cases}
        d_i,&\text{if $J_i$ associates with a real eigenvalue},\\
        d_i/2,&\text{if $J_i$ associates with a pair of conjugate}\\
        &\text{complex eigenvalues,}
    \end{cases}\\
    \mathbb{U}_{ij}&\triangleq
    \begin{cases}
        \{\sum^{i-1}_{\iota=1}d_\iota+1,\cdots,\sum^{i-1}_{\iota=1}d_\iota+j\},\\
        \quad\quad\text{if $J_i$ associates with a real eigenvalue},\\
        \{\sum^{i-1}_{\iota=1}d_\iota+1,\cdots,\sum^{i-1}_{\iota=1}d_\iota+2j\},\\
        \quad\quad\text{if $J_i$ associates with a pair of conjugate}\\
        \quad\quad\text{complex eigenvalues,}
    \end{cases}
\end{aligned}\end{equation*}
where $d_0\triangleq0$ by convention. The following proposition establishes the condition for rendering $(L_\mathrm{e}C, A)$ undetectable. 

\begin{proposition}
\textit{If $[L_\mathrm{e}C]_{\mathbb{U}_{ij}}=0$ for some $j\in\{1,\cdots,\check{d}_i\}$ and $i\in\{1,\cdots,r_\mathrm{u}\}$, then $(L_\mathrm{e}C,A)$ is not detectable.}
\end{proposition}

\begin{proof}
The proof is provided in Appendix A.
\end{proof}

%{\color{red}
%\begin{proof}
%Please refer to the preprint [xx].
%\end{proof}}

Based on the results above, we can construct analytical encryption parameters, as shown in the following corollary.

\begin{corollary}
\textit{Define $\Theta\triangleq [I_{\mathrm{rank}(D)}\ 0](R^{\frac{1}{2}}U)^{-1}$, $[\Theta C]_{\mathbb{U}_{ij}}\overset{\mathrm{svd}}{=} U_{ij}\Sigma_{ij}V_{ij}^\mathrm{T}$, and $\vartheta_{ij}\triangleq\mathrm{rank}([\Theta C]_{\mathbb{U}_{ij}})$, where $U_{ij}\in\mathbb{R}^{\mathrm{ran}(D)\times \mathrm{ran}(D)}$. The coarse privacy is achieved if the parameters $m_c$, $m_e$, and $L_u$ are designed as 
\begin{equation}\begin{aligned}
&m_\mathrm{c}=\min\{\vartheta_{ij},\mathrm{rank}(D)\},\ m_\mathrm{e}=\mathrm{rank}(D)-m_\mathrm{c},\\
&L_\mathrm{u} =\begin{bmatrix}
X_{\mathrm{c}} & X_{12}\\
0 & X_{\mathrm{e}}
\end{bmatrix}U_{ij}^\mathrm{T}\Theta,
\end{aligned}\end{equation}
where $X_{12}\in\mathbb{R}^{m_\mathrm{c}\times m_\mathrm{e}}$ is arbitrary, $X_\mathrm{c}\in\mathbb{R}^{m_\mathrm{c}\times m_\mathrm{c}}$ and $X_\mathrm{e}\in\mathbb{R}^{m_\mathrm{e}\times m_\mathrm{e}}$ can be any invertible matrix, $i\in\{1,\cdots,r_\mathrm{u}\}$, and $j\in\{1,\cdots,\check{d}_i\}$.}
\end{corollary}

\begin{proof}
The proof is provided in Appendix B.
\end{proof}

\begin{remark}
When the system is undetectable, the closed-loop matrix of the estimator cannot be stable. In such cases, it may seem intuitive that the eavesdropper's MSE tends to infinity. However, to the best of the authors' knowledge, no existing literature provides a rigorous proof of this phenomenon. Current literature \cite{Shang9882330,Zou10782997,anderson2005optimal}, and \cite{kailath2000linear} typically demonstrate divergence of the MSE, but lacks a formal proof establishing its unboundedness.
\end{remark}

\begin{remark}
Proposition 1 allows $m_c + m_e$ to range from $\operatorname{rank} D$ to $m$. Corollary 1 specifically selects the minimum value $\operatorname{rank} D$ to minimize the dimension of the transmitted measurement $L_u z$ and thus reduce communication overhead. Although the construction in Corollary 1 yields a particular solution, it retains design freedom through the matrices $X_\mathrm{c}$, $X_\mathrm{e}$, and $X_{12}$, which can be adjusted to meet practical needs without affecting the privacy guarantee.
\end{remark}

\subsection{Intermittent encryption}
This subsection introduces an intermittent strategy to further reduce the encryption consumption. To this end, we first need to analyze the divergence rate of the eavesdropper's MSE. 

\begin{lemma}
\textit{Let the matrices $A\in\mathbb{R}^{n\times n}$, $B\in\mathbb{R}^{n\times l}$, and $C\in\mathbb{R}^{m\times n}$ be partitioned conformally as
\begin{equation*}\begin{aligned}
A=\begin{bmatrix}
A_{11}&A_{12}&A_{13}\\
0&A_{22}&A_{23}\\
0&0&A_{33}\\
\end{bmatrix},\ 
B=\begin{bmatrix}
B_{11}\\
B_{21}\\
0
\end{bmatrix},\ 
C=\begin{bmatrix}
0\\
C_{12}^\mathrm{T}\\
C_{13}^\mathrm{T}
\end{bmatrix}^\mathrm{T},
\end{aligned}\end{equation*}
where $A_{ij}\in\mathbb{R}^{n_i\times n_j}$. 
Let $Q>0$, $R>0$, and $S$ be given matrices of appropriate sizes. Consider the discrete-time Riccati difference equation $X(k+1) = \mathfrak{R}\bigl(X(k);\,A,B,C,Q,S,R\bigr)$ with $X(0)\geq0$. Denote by $X_{11}(k)\in\mathbb{R}^{n_1\times n_1}$ the $n_1\times n_1$ leading principal sub-matrix of $X(k)$. 
If the matrix pair
\begin{equation}
\left(
\begin{bmatrix}
A_{11} & A_{12} \\
0      & A_{22}
\end{bmatrix},\,
\begin{bmatrix}
B_{11} \\ B_{21}
\end{bmatrix}
\right)
\end{equation}
is reachable, then there exist $c>0$ and $N\geq0$ such that $X_{11}(k)\geq c\sum^{k-1}_{i=0}A_{11}^i(A_{11}^i)^{\mathrm{T}}$ for all $k\geq N$.} 
\end{lemma}

\begin{proof}
The proof is provided in Appendix C.
\end{proof}

%{\color{red}
%\begin{proof}
%Please refer to the preprint [xx].
%\end{proof}}

\begin{proposition}
\textit{If $[L_\mathrm{e}C]_{\mathbb{U}_{ij}}=0$ for some $j\in\{1,\cdots,\check{d}_i\}$ and $i\in\{1,\cdots,r_\mathrm{u}\}$, then for some $c>0$,} 
\begin{equation*}
\mathrm{E}\left[\|\tilde{x}(L_\mathrm{e},k)\|^2\right]\geq 
    \begin{cases}
        c\rho(J_i)^{2k}k^{2j-2},\ &\text{if $\rho(J_i)>1$,}\\
        ck^{2j-1},\ &\text{if $\rho(J_i)=1$,}
    \end{cases}\ \forall\ k\geq 0.
\end{equation*}
\end{proposition}

\begin{proof}
The proof is provided in Appendix D.
\end{proof}

A Bernoulli random variable $\gamma(k)$ with $\gamma=\mathrm{Pr}(\gamma(k)=1)$ is introduced to indicate whether $L_\mathrm{c}z(k)$ is encrypted, that is,
\begin{equation}
 \begin{cases}
L_\mathrm{c}z(k)\ \text{is encrypted by (4)}, & \text{if } \gamma(k)=1,\\
L_\mathrm{c}z(k)\ \text{is not encrypted}, & \text{if } \gamma(k)=0.
\end{cases}
\end{equation}
If $\gamma(k)=0$, the eavesdropper has access to $L_\mathrm{u}z(k)$. If $\gamma(k)=1$, the eavesdropper has access only to $L_\mathrm{e}z(k)$. Consequently, the available information set of the eavesdropper can be formulated as $\mathbb{Z}(L_\mathrm{e},\gamma,k)=\{\gamma(0),\cdots,\gamma(k)\}\cup\{\gamma(0)L_\mathrm{e}z(0)+(1-\gamma(0))L_\mathrm{u}z(0),\cdots,\gamma(k)L_\mathrm{e}z(k)+(1-\gamma(k))L_\mathrm{u}z(k)\}$.
Due to the availability of the key $\kappa$, the user is not affected. Thus, to achieve the coarse privacy, we only need to make the eavesdropper's MSE $\mathrm{E}\big[\|x(k+1)-\mathrm{E}[x(k+1)|\mathbb{Z}(L_\mathrm{e},\gamma,k)]\|^2\big]$ approaches infinity. 

\begin{proposition}
\textit{Consider the intermittent encryption (18). If $[L_\mathrm{e}C]_{\mathbb{U}_{ij}} = 0$ for some $j\in\{1,\cdots,\check{d}_i\}$ and $i\in\{1,\cdots,r_\mathrm{u}\}$, and the encryption probability satisfies $\gamma\geq 1/\rho(J_i)^{2}$, then $\lim_{k\to\infty}\mathrm{E}\left[\|x(k+1)-\mathrm{E}[x(k+1)|\mathbb{Z}(L_\mathrm{e},\gamma,k)]\|^2\right]=\infty$.}
\end{proposition}

\begin{proof}
The proof is provided in Appendix E.
\end{proof}

\begin{remark}
The necessity of intermittent encryption arises directly from the polynomial-exponential divergence rate established in Proposition~2. Such a fast rate implies that even if encryption is omitted in some time steps, the divergence may slow to a polynomial or linear rate but still tend to infinity. The key advantages are that it further reduces the encryption overhead compared with partial encryption alone, and it enables more flexible adjustment of the overhead. 
\end{remark}

\section{Precise Privacy}
In this section, we discuss how to protect the confidential information $\varphi^{\mathrm{T}} x_k$ precisely. Before presenting the main results of this section, it is necessary to provide the following lemma. 
\begin{lemma}
\textit{Suppose $X \in \mathbb{R}^{n \times n}$ has no eigenvalue inside the open unit disk. Then, for any nonzero $x\in\mathbb{R}^n$, it holds that $\lim_{k\to\infty}\sum^k_{i=0}x^{\mathrm{T}}X^i(X^i)^{\mathrm{T}}x=\infty$.} 
\end{lemma}

\begin{proof}
The proof is provided in Appendix F.
\end{proof}

\subsection{A necessary and sufficient condition}
\begin{proposition}
\textit{Let $p\in\mathbb{R}^{n}$ be a constant vector. The following two statements are equivalent:
\begin{itemize}
\item $\lim_{k\rightarrow\infty}p^{\mathrm{T}}\mathfrak{L}_k(0,A,B,Q)p<\infty$;
\item $\lim_{k\rightarrow\infty}p^{\mathrm{T}}\mathfrak{L}_k(X,A,B,Q)p<\infty$, $\forall$ $X\geq0$.
\end{itemize}
Define $\mathbb{S}\triangleq\{p\in\mathbb{R}^{n}:\lim_{k\rightarrow\infty}p^{\mathrm{T}}\mathfrak{L}_k(0,A,B,Q)p<\infty\}$, then an analytical representation for $\mathbb{S}$ is given by $\mathbb{S}=\mathrm{Span}(e_{n,n_\mathrm{u}+1},e_{n,n_\mathrm{u}+2},\cdots,e_{n,n})$.}
\end{proposition}

\begin{proof}
The proof is provided in Appendix G.
\end{proof}

With the help of Proposition 4, the following proposition will discuss the existence of $\varphi$-precise privacy. 

\begin{proposition}
\textit{There exist $m_\mathrm{c}$, $m_\mathrm{e}$, and $L_\mathrm{u}$ such that $\lim_{k\rightarrow\infty}\mathrm{E}[\|\varphi^{\mathrm{T}}\tilde{x}(L_\mathrm{e},k)\|^2]=\infty$ if and only if $\varphi\notin\mathbb{S}$.}
\end{proposition}

\begin{proof}
The proof is provided in Appendix H.
\end{proof}

For ease of presentation, partition $\varphi=[\varphi_\mathrm{u};\varphi_\mathrm{s};\varphi_{\bar{\mathrm{c}}}]$, where $\varphi_\mathrm{u}\in\mathbb{R}^{n_\mathrm{u}}$, $\varphi_\mathrm{s}\in\mathbb{R}^{n_\mathrm{s}}$, and $\varphi_{\bar{\mathrm{c}}}\in\mathbb{R}^{n_{\bar{\mathrm{c}}}}$. Then, the following theorem will present a necessary and sufficient condition for $\lim_{k\rightarrow\infty}\mathrm{E}[\|\varphi^{\mathrm{T}}\tilde{x}(L_\mathrm{e},k)\|^2]=\infty$. 

\begin{theorem}
\textit{Consider the system (1). Then, the condition $$\lim_{k\rightarrow\infty}\mathrm{E}[\|\varphi^{\mathrm{T}}\tilde{x}(L_\mathrm{e},k)\|^2]=\infty$$ holds if and only if $\varphi_\mathrm{u}\notin\mathrm{R}(\mathfrak{O}_{n_\mathrm{u}}(L_\mathrm{e}C_\mathrm{u},A_\mathrm{u})^{\mathrm{T}})$.} 
\end{theorem}

\begin{proof}
We first note that if $m_\mathrm{e}=0$, then $\mathrm{R}(\mathfrak{O}_{n_\mathrm{u}}(L_\mathrm{e}C_\mathrm{u},A_\mathrm{u})^{\mathrm{T}})=\{0\}$, and Proposition~4 gives $\lim_{k\to\infty}\mathrm{E}[\|\varphi^{\mathrm{T}}\tilde{x}(L_\mathrm{e},k)\|^2]=\infty$ if and only if $\varphi_\mathrm{u}\neq0$, which is exactly the claimed equivalence. Hence it suffices to consider $m_\mathrm{e}\neq0$ in what follows.

\textit{Necessity.} Suppose $\lim_{k\to\infty}\mathrm{E}[\|\varphi^{\mathrm{T}}\tilde{x}(L_\mathrm{e},k)\|^2]=\infty$. 

Suppose for contradiction that $(L_\mathrm{e}C_\mathrm{u},A_\mathrm{u})$ is observable. In this case, $(L_\mathrm{e}C,A)$ will be detectable (since $\rho(A_\mathrm{s})<1$ and $\rho(A_{\bar{\mathrm{c}}})<1$). Thus, one knows from Kalman filtering theory that $\lim_{k\rightarrow\infty}\mathrm{E}[\|\varphi^{\mathrm{T}}\tilde{x}(L_\mathrm{e},k)\|^2]<\infty$, which is a contradiction. Consequently, $(L_\mathrm{e}C_\mathrm{u},A_\mathrm{u})$ is not observable. 

Perform the observability decomposition via an invertible $T_\mathrm{uo}=[H_{\mathrm{u}\bar{\mathrm{o}}};H_\mathrm{uo}]^{-1}$ such that
\begin{equation*}\begin{aligned}
T_\mathrm{uo}^{-1}A_\mathrm{u}T_\mathrm{uo}=\begin{bmatrix}
A_{\mathrm{u}\bar{\mathrm{o}}}& A_{\mathrm{u}\prime}\\
0&A_\mathrm{uo}
\end{bmatrix},\quad
L_\mathrm{e}C_\mathrm{u}T_\mathrm{uo}=
\begin{bmatrix}
0&C_\mathrm{uo}
\end{bmatrix},
\end{aligned}\end{equation*}
where $(C_\mathrm{uo},A_\mathrm{uo})$ is observable, $A_{\mathrm{u}\bar{\mathrm{o}}}\in\mathbb{R}^{n_{\mathrm{u}\bar{\mathrm{o}}}\times n_{\mathrm{u}\bar{\mathrm{o}}}}$, $A_\mathrm{uo}\in\mathbb{R}^{n_\mathrm{uo}\times n_\mathrm{uo}}$, $C_\mathrm{uo}\in\mathbb{R}^{m_\mathrm{e}\times n_\mathrm{uo}}$, $n_{\mathrm{u}\bar{\mathrm{o}}}+n_\mathrm{uo}=n_\mathrm{u}$, $n_\mathrm{uo}=\mathrm{rank}(\mathfrak{O}_{n_\mathrm{u}}(L_\mathrm{e}C_\mathrm{u},A_\mathrm{u}))$, the rows of $H_\mathrm{uo}\in\mathrm{R}^{n_\mathrm{uo}\times n_\mathrm{u}}$ are a basis of the row space of $\mathfrak{O}_{n_\mathrm{u}}(L_\mathrm{e}C_\mathrm{u},A_\mathrm{u})$ and the rows of $H_{\mathrm{u}\bar{\mathrm{o}}}\in\mathrm{R}^{n_{\mathrm{u}\bar{\mathrm{o}}}\times n_\mathrm{u}}$ are linearly independent of each other and also linearly independent of the rows of $H_\mathrm{uo}$. Then, one can use the linear transformation $T=T_\mathrm{uo}\oplus I$ to obtain 
\begin{align}
\check{A}&=T^{-1}AT=\begin{bmatrix}\begin{array}{c:ccc}
A_{\mathrm{u}\bar{\mathrm{o}}}&A_{\mathrm{u}\prime}&0&A_{\prime11}\\
\hdashline
0&A_\mathrm{uo}&0&A_{\prime12}\\
0&0&A_\mathrm{s}&A_{\prime2}\\
0&0&0&A_{\bar{\mathrm{c}}}
\end{array}\end{bmatrix}\nonumber\\
&\quad\quad\quad\quad\ \ =\begin{bmatrix}
A_{\mathrm{u}\bar{\mathrm{o}}}&A_{\mathrm{o}12}\\
0&A_\mathrm{o}
\end{bmatrix},\\
\check{B}&=T^{-1}B=\begin{bmatrix}\begin{array}{c}
B_{\mathrm{u}\bar{\mathrm{o}}}\\
\hdashline
B_\mathrm{uo}\\
B_\mathrm{s}\\
0\\
\end{array}\end{bmatrix}=\begin{bmatrix}
B_{\mathrm{u}\bar{\mathrm{o}}}\\
B_\mathrm{o}
\end{bmatrix},\nonumber\\
\check{C}_\mathrm{e}&=L_\mathrm{e}CT=\begin{bmatrix}
\begin{array}{c:ccc}
0 & C_\mathrm{uo} & L_\mathrm{e}C_\mathrm{s} & L_\mathrm{e}C_{\bar{\mathrm{c}}}
\end{array}\end{bmatrix}
=\begin{bmatrix}
0 & C_\mathrm{o}
\end{bmatrix}.\nonumber
\end{align}

Set $\check{\varphi}=[\check{\varphi}_1;\check{\varphi}_2;\check{\varphi}_3]\triangleq T^{\mathrm{T}}\varphi$ with $\check{\varphi}_1\in\mathbb{R}^{n_{\mathrm{u}\bar{\mathrm{o}}}}$, $\check{\varphi}_2\in\mathbb{R}^{n_{\mathrm{uo}}}$, $\check{\varphi}_3\in\mathbb{R}^{n-n_\mathrm{u}}$. Then, one has 
\begin{equation*}\begin{aligned}
\varphi=T^{-\mathrm{T}}
\begin{bmatrix}
\check{\varphi}_1\\
\check{\varphi}_{2}\\
\check{\varphi}_{3}
\end{bmatrix}=\begin{bmatrix}
T_\mathrm{uo}^{-\mathrm{T}}&0\\
0&I
\end{bmatrix}\begin{bmatrix}
{\begin{bmatrix}
\check{\varphi}_1\\
\check{\varphi}_2
\end{bmatrix}}\\
\check{\varphi}_3
\end{bmatrix}.
\end{aligned}\end{equation*}
This implies that 
\begin{equation*}\begin{aligned}
\varphi_\mathrm{u}=\begin{bmatrix}
H^{\mathrm{T}}_{\mathrm{u}\bar{\mathrm{o}}}&H_\mathrm{uo}^{\mathrm{T}}
\end{bmatrix}\begin{bmatrix}
\check{\varphi}_1\\
\check{\varphi}_{2}
\end{bmatrix}.
\end{aligned}\end{equation*}
By the unique representation in a basis, $\varphi_\mathrm{u}\notin\mathrm{R}(\mathfrak{O}_{n_\mathrm{u}}(L_\mathrm{e}C_\mathrm{u},A_\mathrm{u})^{\mathrm{T}})$ if and only if $\check{\varphi}_1\neq0$. Hence we need to show that $\lim_{k\to\infty}\mathrm{E}[\|\varphi^{\mathrm{T}}\tilde{x}(L_\mathrm{e},k)\|^2]=\infty$ implies $\check{\varphi}_1\neq0$. We prove its contrapositive: $\check{\varphi}_1=0$ implies $\lim_{k\to\infty}\mathrm{E}[\|\varphi^{\mathrm{T}}\tilde{x}(L_\mathrm{e},k)\|^2]<\infty$.

Define $\hat{\Xi}(L_\mathrm{e},k)\triangleq T^{-1}\hat{P}(L_\mathrm{e},k)T^{-\mathrm{T}}$, which satisfies the Riccati difference equation 
\begin{equation*}\begin{aligned}
\hat{\Xi}(L_\mathrm{e},k+1)=\mathfrak{R}\big(\hat{\Xi}(L_\mathrm{e},k);\check{A},\check{B},\check{C}_\mathrm{e},Q,S,L_\mathrm{e}RL_\mathrm{e}^\mathrm{T}\big).
\end{aligned}\end{equation*}
Partition $\hat{\Xi}(L_\mathrm{e},k)$ as
\begin{equation}\begin{aligned}
\hat{\Xi}(L_\mathrm{e},k)=\begin{bmatrix}
\hat{\Xi}_{11}(L_\mathrm{e},k)&\hat{\Xi}_{12}(L_\mathrm{e},k)\\
\hat{\Xi}_{21}(L_\mathrm{e},k)&\hat{\Xi}_{22}(L_\mathrm{e},k)
\end{bmatrix},
\end{aligned}\end{equation}
where $\hat{\Xi}_{11}(L_\mathrm{e},k)\in\mathbb{R}^{n_{\mathrm{u}\bar{\mathrm{o}}}\times n_{\mathrm{u}\bar{\mathrm{o}}}}$ and $\hat{\Xi}_{22}(L_\mathrm{e},k)\in\mathbb{R}^{(n-n_{\mathrm{u}\bar{\mathrm{o}}})\times(n-n_{\mathrm{u}\bar{\mathrm{o}}})}$. From (19) one obtains
\begin{equation*}\begin{aligned}
\check{A}\hat{\Xi}(L_\mathrm{e},k)\check{C}_\mathrm{e}^\mathrm{T}+\check{B}S
&=\begin{bmatrix}
\star\\ A_\mathrm{o}\hat{\Xi}_{22}(L_\mathrm{e},k)C_\mathrm{o}^\mathrm{T}+B_\mathrm{o}S
\end{bmatrix},\\
\check{C}_\mathrm{e}\hat{\Xi}(L_\mathrm{e},k)\check{C}_\mathrm{e}^\mathrm{T}+L_\mathrm{e}RL_\mathrm{e}^\mathrm{T}
&=C_\mathrm{o}\hat{\Xi}_{22}(L_\mathrm{e},k)C_\mathrm{o}^\mathrm{T}+L_\mathrm{e}RL_\mathrm{e}^\mathrm{T},
\end{aligned}\end{equation*}
where the symbol $\star$ denotes sub-matrices that are irrelevant to the proof. Consequently, one can verify that 
\begin{equation*}\begin{aligned}
\hat{\Xi}_{22}(L_\mathrm{e},k+1)=\mathfrak{R}\big(\hat{\Xi}_{22}(L_\mathrm{e},k);&A_\mathrm{o},B_\mathrm{o},C_\mathrm{o},Q,S,L_\mathrm{e}RL_\mathrm{e}^\mathrm{T}\big).
\end{aligned}\end{equation*}

Since $(C_\mathrm{uo},A_\mathrm{uo})$ is observable, $\rho(A_{\mathrm{s}})<1$, and $\rho(A_{\bar{\mathrm{c}}})<1$, the pair $(C_\mathrm{o},A_\mathrm{o})$ is detectable, so Kalman filtering theory yields $\lim_{k\to\infty}\hat{\Xi}_{22}(L_\mathrm{e},k)<\infty$. If $\check{\varphi}_1=0$, then
\begin{equation*}\begin{aligned}
\lim_{k\to\infty}\mathrm{E}[\|\varphi^{\mathrm{T}}\tilde{x}(L_\mathrm{e},k)\|^2]&=\lim_{k\to\infty}\check{\varphi}^{\mathrm{T}}\hat{\Xi}(L_\mathrm{e},k+1)\check{\varphi}\\
&=\lim_{k\to\infty}
\begin{bmatrix}
\check{\varphi}_2\\ \check{\varphi}_3
\end{bmatrix}^{\mathrm{T}}
\hat{\Xi}_{22}(L_\mathrm{e},k+1)
\begin{bmatrix}
\check{\varphi}_2\\ \check{\varphi}_3
\end{bmatrix}\\
&<\infty,
\end{aligned}\end{equation*}
which completes the proof of necessity.

\textit{Sufficiency.} Suppose $\varphi_\mathrm{u}\notin\mathrm{R}(\mathfrak{O}_{n_\mathrm{u}}(L_\mathrm{e}C_\mathrm{u},A_\mathrm{u})^{\mathrm{T}})$.

In this case, $(L_\mathrm{e}C_\mathrm{u},A_\mathrm{u})$ cannot be observable. Then, similar to the proof of necessity, we again have $\varphi_u\notin\mathrm{R}(\mathcal{O}_{n_\mathrm{u}}(L_\mathrm{e}C_\mathrm{u},A_\mathrm{u})^{\mathrm{T}})$ if and only if $\check{\varphi}_1\neq0$. Thus, it remains to show that $\lim_{k\to\infty}\check{\varphi}^{\mathrm{T}}\hat{\Xi}(L_\mathrm{e},k+1)\check{\varphi}<\infty$ implies $\check{\varphi}_1=0$, i.e., the contrapositive. 

Suppose $\lim_{k\to\infty}\check{\varphi}^{\mathrm{T}}\hat{\Xi}(L_\mathrm{e},k)\check{\varphi}<\infty$. By the invariance of reachability under similarity transformations, the pair
\begin{equation*}\begin{aligned}
\left(\begin{bmatrix}
A_{\mathrm{u}\bar{\mathrm{o}}}&A_{\mathrm{u}\prime}&0\\
0&A_\mathrm{uo}&0\\
0&0&A_\mathrm{s}
\end{bmatrix},\,
\begin{bmatrix}
B_{\mathrm{u}\bar{\mathrm{o}}}\\ B_\mathrm{uo}\\ B_\mathrm{s}
\end{bmatrix}\right)
\end{aligned}\end{equation*}
is reachable. Then Lemma 1 guarantees the existence of $c>0$ and $N\ge 0$ such that
\begin{equation}\begin{aligned}
\hat{\Xi}_{11}(L_\mathrm{e},k)\ge c\sum^{k-1}_{i=0}A_{\mathrm{u}\bar{\mathrm{o}}}^i(A_{\mathrm{u}\bar{\mathrm{o}}}^i)^{\mathrm{T}},\quad\forall\,k\ge N.
\end{aligned}\end{equation}

Additionally, it has been demonstrated in the analysis of the necessity that $\lim_{k\to\infty}\hat{\Xi}_{22}(L_\mathrm{e},k+1)<\infty$. Then, we can derive that 
\begin{equation*}\begin{aligned}
&\mathrel{\phantom{=}}\lim_{k\to\infty}\check{\varphi}_1^{\mathrm{T}}\hat{\Xi}_{11}(L_\mathrm{e},k+1)\check{\varphi}_1\\
%&=\lim_{k\to\infty}[\check{\varphi}_1;0]^{\mathrm{T}}\hat{\Xi}(L_e,k+1)[\check{\varphi}_1;0]\\
&=\lim_{k\to\infty}\left(\check{\varphi}-
\begin{bmatrix}
0\\
\check{\varphi}_2\\
\check{\varphi}_3
\end{bmatrix}\right)^{\mathrm{T}}\hat{\Xi}(L_\mathrm{e},k+1)\left(\check{\varphi}-
\begin{bmatrix}
0\\
\check{\varphi}_2\\
\check{\varphi}_3
\end{bmatrix}\right)\\
%&\leq2\lim_{k\to\infty}\bar{\varphi}^{\mathrm{T}}\hat{\Xi}(L_e,k+1)\bar{\varphi}+[0;\bar{\varphi}_2]^{\mathrm{T}}\hat{\Xi}(L_e,k+1)[0;\bar{\varphi_2}]\\
&\leq2\lim_{k\to\infty}\check{\varphi}^{\mathrm{T}}\hat{\Xi}(L_\mathrm{e},k+1)\check{\varphi}+2
\begin{bmatrix}
\check{\varphi}_2\\
\check{\varphi}_3
\end{bmatrix}^{\mathrm{T}}\hat{\Xi}_{22}(L_\mathrm{e},k+1)
\begin{bmatrix}
\check{\varphi}_2\\
\check{\varphi}_3
\end{bmatrix}\\
&<\infty.
\end{aligned}\end{equation*}
This together with the lower bound (21) implies
\begin{equation*}\begin{aligned}
\lim_{k\to\infty} \sum^{k-1}_{i=0}\check{\varphi}_1^{\mathrm{T}}A_{\mathrm{u}\bar{\mathrm{o}}}^i(A_{\mathrm{u}\bar{\mathrm{o}}}^i)^{\mathrm{T}}\check{\varphi}_1<\infty.
\end{aligned}\end{equation*}
Since all eigenvalues of $A_{\mathrm{u}\bar{\mathrm{o}}}$ are unstable, Lemma~2 forces $\check{\varphi}_1=0$. This completes the proof of sufficiency.
\end{proof}

\begin{remark} Proposition 5 and Theorem 3 demonstrate that achieving divergent MMSE estimates for $\varphi^{\mathrm{T}}x_k$ essentially requires the unstable part of $\varphi$ to lie outside the system's observable subspace. These results align with our intuition, as driving the estimation error to approach infinity inherently relies on the unstable mode of the system. 
%If the confidential direction contains no unstable parts (Theorem 3) or if its unstable parts are observed by the eavesdropper (Theorem 4), the estimation error in that direction remains bounded. 
\end{remark}

\subsection{Eliminating vectors from the observable subspace}
%Based on Theorems 1, 4, and 5, one knows that for $\varphi\notin\mathbb{S}$, then $\varphi$-precise privacy is achieved if $m_\mathrm{c}$, $m_\mathrm{e}$, and $L_\mathrm{u}$ satisfy 
%\begin{equation}
%L_\mathrm{u}=\begin{bmatrix}
%L_\mathrm{c}\\
%L_\mathrm{e}
%\end{bmatrix}\Theta,\ \varphi_\mathrm{u}\notin\mathrm{R}(\mathfrak{O}_{n_\mathrm{u}}(X_\mathrm{e}\Theta C_\mathrm{u},A_\mathrm{u})^{\mathrm{T}}),
%\end{equation}
%where $L_\mathrm{c}\in\mathbb{R}^{m_\mathrm{c}\times\mathrm{rank}(D)}$, $L_\mathrm{e}\in\mathbb{R}^{m_\mathrm{e}\times\mathrm{rank}(D)}$, $m_\mathrm{c}+m_\mathrm{e}=\mathrm{rank}(D)$, and $X_\mathrm{c}$ can be any full-row-rank matrix whose rows are linearly independent of each other and also linearly independent of the rows of $L_\mathrm{e}$. 

In this subsection, we will present a method for eliminating arbitrary vectors from a given observable subspace. 

\begin{lemma}
\textit{Let $ A\in\mathbb{R}^{n \times n} $ and $C\in\mathbb{R}^{m \times n} $ be matrices without any partition structure. Given a nonzero vector $\varphi\in\mathbb{R}^n $. Then, $\varphi\notin\mathrm{R}(\mathfrak{O}_n(C,A)^{\mathrm{T}})$ if and only if there exists a vector $\beta\in\mathbb{R}^n$ such that}
\begin{enumerate}
    \item $ CA^i \beta = 0 $ for $i=0,1,\cdots,n-1$,
    \item $ \varphi^{\mathrm{T}}\beta\neq 0 $.
\end{enumerate}
\end{lemma}

\begin{proof}
The proof is provided in Appendix I.
\end{proof}

%\begin{proposition}
%Given a nonzero vector $\varphi\notin\mathbb{S}$. Then, $\varphi_\mathrm{u}\notin\mathrm{R}(\mathfrak{O}_n(X_\mathrm{e}\Theta C_\mathrm{u},A_\mathrm{u})^{\mathrm{T}})$ if and only if there exists a vector $\beta\in\mathbb{R}^n$ such that 
%\begin{enumerate}
%    \item $ X_\mathrm{e}\Theta C_\mathrm{u}A_\mathrm{u}^i \beta = 0 $ for $i=0,1,\ldots,n-1$,
%    \item $ \varphi_\mathrm{u}^{\mathrm{T}}\beta\neq 0 $.
%\end{enumerate}
%\end{proposition}
%
%\begin{proof}
%It follows from the orthogonal complement principle that $\mathrm{R}(\mathfrak{O}_n(C,A)^\mathrm{T})=\mathrm{N}(\mathfrak{O}_n(C, A))^\perp$. Thus, $ \varphi\notin\mathrm{R}(\mathfrak{O}_n(C, A)^\mathrm{T})$ if and only if there exists $\beta \in \mathrm{N}(\mathfrak{O}_n(C, A))$ such that $ \varphi^\mathrm{T}\beta\neq0$. Meanwhile, the first condition is equivalent to $\beta \in \mathrm{N}(\mathfrak{O}_n(C, A))$. The proof is completed. 
%\end{proof}

For ease of presentation, partition $\varphi_\mathrm{u}=[\varphi_{\mathrm{u},1},\varphi_{\mathrm{u},2},\cdots,\varphi_{\mathrm{u},r_\mathrm{u}}]$ and $\varphi_{\mathrm{u},i}=[\varphi_{\mathrm{u},i,1};\varphi_{\mathrm{u},i,2};\cdots;\varphi_{\mathrm{u},i,\check{d}_i}]$, where $\varphi_{\mathrm{u},i}\in\mathbb{R}^{d_i}$ and $\varphi_{\mathrm{u},i,j}\in\mathbb{R}^{d_i/\check{d}_i}$. In the following corollary, we construct a family of analytical linear transformations that achieve $\varphi$-precise privacy. 

\begin{corollary}
\textit{If $\varphi_{\mathrm{u},i,j}\neq 0$ for some $i\in\{1,\cdots,r_\mathrm{u}\}$ and $j\in\{1,\cdots,\check{d}_i\}$, the $\varphi$-precise privacy is achieved if the parameters $m_c$, $m_e$, and $L_u$ are designed as (16).}
%\begin{align*}
%&m_\mathrm{c}=\min\{|\mathbb{U}_i(j)|,\mathrm{rank}(D)\},\ m_\mathrm{e}=\mathrm{rank}(D)-m_\mathrm{c},\nonumber\\
%&L_\mathrm{u} =\begin{bmatrix}
%X_{\mathrm{c}} & X_{12}\\
%0 & X_{\mathrm{e}}
%\end{bmatrix}U_{i(j)}^\mathrm{T}\Theta,
%\end{align*}
%where $X_{12}\in\mathbb{R}^{m_\mathrm{c}\times m_\mathrm{e}}$ is arbitrary, $X_\mathrm{c}\in\mathbb{R}^{m_\mathrm{c}\times m_\mathrm{c}}$ and $X_\mathrm{e}\in\mathbb{R}^{m_\mathrm{e}\times m_\mathrm{e}}$ can be any invertible matrix.
\end{corollary}

\begin{proof}
The proof is provided in Appendix J.
\end{proof}

\begin{remark}
Information in dynamical systems exhibits a hierarchical dependency: protecting a specific quantity (e.g., velocity) requires concealing not only its direct measurements but also those of its causal predecessors (e.g., position), as the latter can be used to infer the former. This depth of dependency is captured by the index $j$ in the nonzero block $\varphi_{\mathrm{u},i,j}$. A larger $j$ indicates that the confidential information lies deeper in the causal chain, which in turn requires encrypting more measurement channels to sever the information flow ($\vartheta_{ij}\leq\vartheta_{i\ell}$ if $j\leq \ell$). Moreover, according to Proposition 2, a deeper information depth (larger $j$) also results in a faster divergence rate for the eavesdropper’s estimation error.
\end{remark}

\section{Simulations}
\subsection{Second-order differential system}
Coarse privacy is examined here, while its comparison with precise privacy is presented in the next subsection. 
%The simulation model follows subsection V-A of the preprint \cite{Hu9705534}.

Consider a classical damped mass-spring system governed by the second-order differential equation
\begin{equation*}
\begin{aligned}
\dot{x} = A x + B w,
\end{aligned}
\end{equation*}
where
\begin{equation*}
\begin{aligned}
A &= \begin{bmatrix}
0 & 0 & 1 & 0 \\
0 & 0 & 0 & 1 \\
-\frac{k_1}{m_1} & \frac{k_1}{m_1} & -\frac{c_1}{m_1} & \frac{c_1}{m_1} \\[2pt]
\frac{k_1}{m_2} & -\frac{k_1+k_2}{m_2} & \frac{c_1}{m_2} & -\frac{c_1+c_2}{m_2}
\end{bmatrix},\\
B &= \begin{bmatrix}
0 & 0  \\
0 & 0 \\
\frac{1}{m_1} & 0  \\
0 & \frac{1}{m_2}
\end{bmatrix}. 
\end{aligned}
\end{equation*}
The system parameters are set to $k_1=20$, $k_2=1$, $c_1=-2$, $c_2=-1$, $m_1=1$, and $m_2=2$.  Discretizing with a sampling period of $0.1$ yields the discrete-time state equation
\begin{equation*}
\begin{aligned}
x(k+1)=A_{\mathrm{dis}}x(k)+B_{\mathrm{dis}}w(k),
\end{aligned}
\end{equation*}
where 
\begin{equation*}
\begin{aligned}
&A_{\mathrm{dis}}=\begin{bmatrix}
0.8920 & 0.1082 & 0.1073 & -0.0074 \\
0.0549 & 0.9425 & -0.0037 & 0.1062 \\
-2.2191 & 2.2228 & 1.1139 & -0.1174 \\
1.1361 & -1.1892 & -0.0587 & 1.1092
\end{bmatrix}, \\
&B_{\mathrm{dis}}=\begin{bmatrix}
0.0053 & -0.0001 \\
-0.0001 & 0.0026 \\
0.1073 & -0.0037 \\
-0.0037 & 0.0531
\end{bmatrix}.
\end{aligned}
\end{equation*}
Moreover, the measurement matrix is $C = [e_{4,1}, e_{4,2}, e_{4,1}, e_{4,2}, e_{4,3}, e_{4,4}]^{\mathrm{T}}$, and the noise covariances are chosen as $S=0$, $Q=I$, and $R=0.25I$.

Fig.~1 compares the MSE of the user and the eavesdropper under the proposed coarse privacy-preserving method. The user's MSE matches the optimal Kalman filter, while the eavesdropper's total MSE diverges to infinity, confirming that coarse privacy is achieved. Fig.~1 also shows that the intermittent encryption scheme (with $\gamma=0.75$) still causes divergence, albeit at a slower growth rate.

To highlight the advantages of our method, we benchmark it against the coarse privacy-preserving method in \cite{Shang9882330}. Both methods drive the eavesdropper's total MSE to infinity, but Table~1 reveals that our method achieves reductions in computational complexity and communication overhead. The source of these improvements lies in a structural difference: the considered system possesses only complex eigenvalues, so the scheme in \cite{Shang9882330} must operate in the complex vector space, handling both real and imaginary parts, which increases both computational burden and data dimension. In contrast, our method performs lossless compression entirely in the real vector space, eliminating the need for complex arithmetic and resulting in a more efficient implementation.

\begin{figure}[b]
      \centering
      \includegraphics[scale=0.45]{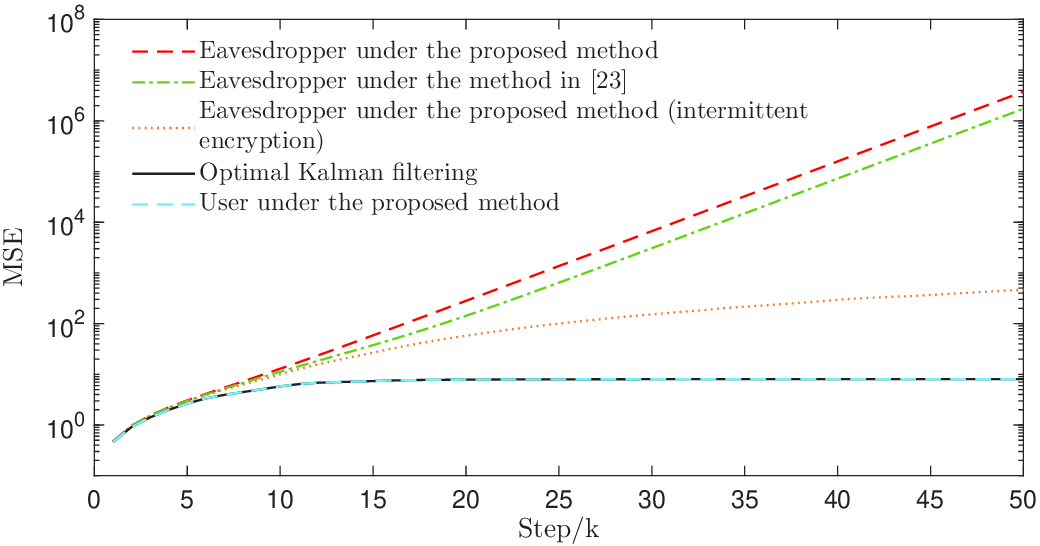}
      \caption{The MSEs of eavesdropper under the proposed coarse privacy-preserving method and the coarse privacy-preserving method in \cite{Shang9882330}.}
\end{figure}

\subsection{Target tracking system}
This subsection demonstrates precise privacy and highlights its key advantage over coarse privacy. 
%All simulation parameters follow subsection V-B of the preprint \cite{Hu9705534}.

Consider a target tracking scenario in three-dimensional space, where the system matrices are given by
\begin{equation*}
\left\{ \begin{array}{l}
A = \mathrm{diag}(A_x, A_y, A_z),\quad B = I,\\[2pt]
C = [e_{9,1}, e_{9,3}, e_{9,4}, e_{9,6}, e_{9,7}, e_{9,9}]^{\mathrm{T}},\\[2pt]
A_x = A_y = A_z = \begin{bmatrix}
1 & \Delta t & \frac{\Delta t^2}{2} \\
0 & 1 & \Delta t \\
0 & 0 & 1
\end{bmatrix},\ \Delta t = 0.1.
\end{array} \right.
\end{equation*}
The noise covariances are chosen as $Q = \mathrm{diag}(Q_x, Q_y, Q_z)$, $R = \mathrm{diag}(1, 0.04, 1, 0.04, 1, 0.04)$, and $S = 0$, where
\begin{equation*}\begin{aligned}
Q_\iota = S_\iota \begin{bmatrix}
\frac{\Delta t^5}{20} & \frac{\Delta t^4}{8} & \frac{\Delta t^3}{6} \\[2pt]
\frac{\Delta t^4}{8} & \frac{\Delta t^3}{3} & \frac{\Delta t^2}{2} \\[2pt]
\frac{\Delta t^3}{6} & \frac{\Delta t^2}{2} & \Delta t
\end{bmatrix},\quad \iota \in \{x, y, z\},
\end{aligned}\end{equation*}
with $S_x = S_y = 1$ and $S_z = 0.25$.

The confidential information is defined as the acceleration in the $z$-direction, i.e., $\varphi=e_{9,9}$. Fig.~2 compares the directional MSE of this confidential variable under three schemes: the coarse privacy-preserving method in \cite{Shang9882330}, the proposed coarse privacy-preserving method, and the proposed precise privacy-preserving method. It can be observed from Fig.~2 that under either coarse privacy-preserving method, the eavesdropper's directional MSE remains bounded, meaning that the confidential $z$-acceleration can still be accurately estimated and is therefore leaked. In sharp contrast, under the proposed precise privacy-preserving method, the directional MSE diverges to infinity. %To further illustrate this, Fig.~3 displays the eavesdropper's actual estimates of the $z$-acceleration under the three schemes. With either coarse privacy method, the eavesdropper faithfully reconstructs the confidential signal, whereas under precise privacy the estimate becomes completely unreliable. 
These results validate one of the core contributions of this work, namely that precise privacy provides directional protection that coarse privacy inherently cannot provide. 

\begin{table}
\centering
\caption{Comparison of the proposed method with the method in \cite{Shang9882330} in terms of computation and communication}
\begin{tabular}{|l|l|l|} 
\hline
                          & The proposed method & The method in \cite{Shang9882330}  \\ 
\hline
Computation time          &  $2.23\times 10^{-4}$s                   & $7.20\times 10^{-4}$s                  \\ 
\hline
Communication cost & $4$ scalars                   & $2\times 6$ scalars                  \\ 
\hline
Encryption cost    & $2(\times 0.75)$ scalars                   & $2$ scalars                  \\
\hline
\end{tabular}
\end{table}

\begin{figure}
      \centering
      \includegraphics[scale=0.45]{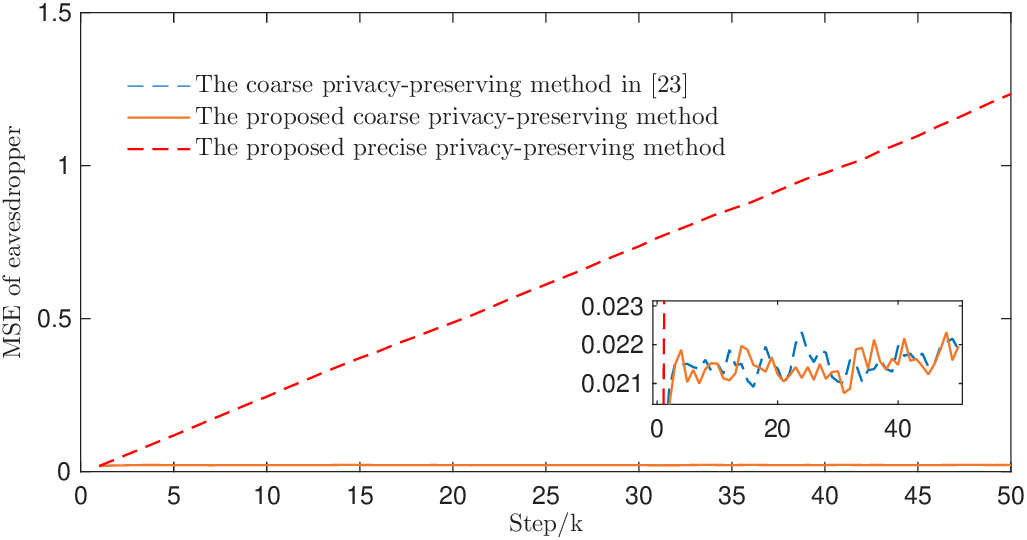}
      \caption{The MSEs of eavesdropper on the confidential direction under the coarse and precise privacy-preserving methods.}
\end{figure}

\section{Conclusion}
This paper proposed a directional encryption methodology for achieving both coarse and precise privacy in state estimation. For coarse privacy, an analytical linear transformation was derived to achieve lossless compression and destroy the eavesdropper's detectability. Furthermore, the growth rate of the eavesdropper's MSE was shown to follow a polynomial-exponential form governed by the encrypted unstable modes. For precise privacy, we proved that a confidential direction was protected if and only if its unstable component lies outside the eavesdropper's observable subspace, and we developed a systematic method to exclude arbitrary target vectors from that subspace. 

\appendix
\subsection{The proof of Proposition 1}
\begin{proof}
Suppose that the condition in the proposition holds, one can derive that 
\begin{equation*}\begin{aligned}
[KL_\mathrm{e}C]_{\mathbb{U}_{ij}}=K[L_\mathrm{e}C]_{\mathbb{U}_{ij}}=0\Rightarrow K[L_\mathrm{e}C]_{\mathbb{U}_{i1}}=0
\end{aligned}\end{equation*}
for any $K$ with appropriate dimension. Then, recall the partition structure of $A$ shown in (3), one has 
\begin{equation}\begin{aligned}
A+KL_\mathrm{e}C=
\begin{bmatrix}
\star & 0  &\star  \\
\star & \Lambda_i          &\star  \\
\star & 0          &\star
\end{bmatrix},
\end{aligned}\end{equation}
where the symbol $\star$ denotes some matrices that do not contribute to the proof. It can be derived from (22) that $\rho(A+KL_\mathrm{e}C)\geq\rho(\Lambda_i)\geq1$. Consequently, $(L_\mathrm{e}C,A)$ is not detectable. The proof is completed. 
\end{proof}

\subsection{The proof of Corollary 1}
\begin{proof}
It follows from Theorem 1 that $\hat{x}(L_\mathrm{u},k)=\hat{x}(I,k)$ and $\hat{P}(L_\mathrm{u},k)=\hat{P}(I,k)$ for all $k\geq1$ if the parameters satisfy (16). 

When $\vartheta_{ij}\geq \mathrm{rank}(D)$, the selection (16) gives $m_e=0$. Under such case, one can derive that 
\begin{equation*}\begin{aligned}
\lim_{k\to\infty}\mathrm{E}[\|\tilde{x}(L_\mathrm{e},k)\|^2]=\lim_{k\to\infty}\mathrm{Tr}(\mathfrak{L}_k(P_0,A,B,Q))=\infty.
\end{aligned}\end{equation*}

When $\vartheta_{ij}<\mathrm{rank}(D)$, the selection (16) gives $L_\mathrm{e}=X_\mathrm{e}([U_{ij}]_{\{\vartheta_{ij}+1,\cdots,\mathrm{rank}(D)\}})^{\mathrm{T}}\Theta$. Then, we have 
\begin{equation*}\begin{aligned}
[L_\mathrm{e}C]_{\mathbb{U}_{ij}}&=X_\mathrm{e}([U_{ij}]_{\{\vartheta_{ij}+1,\cdots,\mathrm{rank}(D)\}})^{\mathrm{T}}[\Theta C]_{\mathbb{U}_{ij}}\\
&=X_\mathrm{e}([U_{ij}]_{\{\vartheta_{ij}+1,\cdots,\mathrm{rank}(D)\}})^{\mathrm{T}}U_{ij}\Sigma_{ij}V_{ij}^\mathrm{T}\\
&=0,
\end{aligned}\end{equation*}
where the last equality follows from the orthogonality of $U_{ij}$. Thus, the condition of Proposition~1 is satisfied, and by Theorem~2 we have $\lim_{k\rightarrow\infty}\mathrm{E}\left[\|\tilde{x}(L_\mathrm{e},k)\|^2\right]=\infty$. The proof is completed.
\end{proof}

\subsection{The proof of Lemma 1}
\begin{proof}
Let \(n_\diamond = n_1+n_2\). Iterating the Riccati difference equation \(n_\diamond\) times gives
\begin{equation}
X(k n_\diamond) = \mathfrak{R}\bigl(X((k-1)n_\diamond);
\mathcal{A},\mathcal{B},\mathcal{C},\mathcal{Q},\mathcal{S},\mathcal{R}\bigr), 
\end{equation}
where
\begin{equation*}\begin{aligned}
\mathcal{A} &= A^{n_\diamond}, \quad
\mathcal{B} = \mathfrak{C}_{n_\diamond}(A,B), \quad
\mathcal{C} = \mathfrak{O}_{n_\diamond}(A,C),  \\
\mathcal{D} &= \mathfrak{T}_{n_\diamond}(A,B,C,0_{m\times l}), \quad\quad
\mathcal{Q} = I_{n_\diamond}\otimes Q,\\
\mathcal{R} &=
\begin{bmatrix} \mathcal{D} & I \end{bmatrix}
\begin{bmatrix}
\mathcal{Q} & I_{n_\diamond}\otimes S \\
(I_{n_\diamond}\otimes S)^{\mathrm{T}} & I_{n_\diamond}\otimes R
\end{bmatrix}
\begin{bmatrix} \mathcal{D}^{\mathrm{T}} \\ I \end{bmatrix},  \\
\mathcal{S} &=
\begin{bmatrix} \mathcal{Q} & I_{n_\diamond}\otimes S \end{bmatrix}
\begin{bmatrix} \mathcal{D} & I \end{bmatrix}^{\mathrm{T}}. 
\end{aligned}\end{equation*}
Introduce the transformed matrices
\begin{equation*}\begin{aligned}
\check{\mathcal{A}} &\triangleq \mathcal{A} - \mathcal{B}\mathcal{S}\mathcal{R}^{-1}\mathcal{C},\\ 
\check{\mathcal{Q}} &\triangleq \mathcal{Q} - \mathcal{S}\mathcal{R}^{-1}\mathcal{S}^{\mathrm{T}}.
\end{aligned}\end{equation*}
Then (23) can be rewritten as
\begin{equation}
X(k n_\diamond) = \mathfrak{R}\bigl(X((k-1)n_\diamond);\,
\check{\mathcal{A}},\mathcal{B},\mathcal{C},\check{\mathcal{Q}},0,\mathcal{R}\bigr).
\end{equation}

Notice that \(\check{\mathcal{Q}}\) is the Schur complement of
\begin{equation*}
\begin{bmatrix} I & 0 \\ \mathcal{D} & I \end{bmatrix}
\begin{bmatrix}
\mathcal{Q} & I_{n_\diamond}\otimes S \\
(I_{n_\diamond}\otimes S)^{\mathrm{T}} & I_{n_\diamond}\otimes R
\end{bmatrix}
\begin{bmatrix} I & 0 \\ \mathcal{D} & I \end{bmatrix}^{\mathrm{T}} > 0,
\end{equation*}
hence \(\check{\mathcal{Q}} > 0\). Since the pair in (17) is reachable, there exists \(c_1>0\) such that
\begin{equation*}
\mathcal{B}\mathcal{Q}\mathcal{B}^{\mathrm{T}} \ge c_1 (I_{n_1}\oplus 0_{n_2+n_3}). 
\end{equation*}
Applying the monotonicity of the Riccati equation, we obtain
\begin{equation}
X(k n_\diamond) \ge \mathfrak{R}_k\bigl(0;\,
\check{\mathcal{A}},\mathcal{B},\mathcal{C},c_1(I_{n_1}\oplus 0_{n_2+n_3}),0,\mathcal{R}\bigr).
\end{equation}

We now analyze the right-hand side of (25). First,
\begin{equation*}
\mathfrak{R}_1\bigl(0;\,
\check{\mathcal{A}},\mathcal{B},\mathcal{C},c_1(I_{n_1}\oplus 0_{n_2+n_3}),0,\mathcal{R}\bigr)
= c_1(I_{n_1}\oplus 0_{n_2+n_3}). 
\end{equation*}
From the block structure of \(C\) and \(\mathcal{C}\), it follows that
\begin{equation*}
\mathfrak{R}_1\bigl(0;\check{\mathcal{A}},I,\mathcal{C},c_1(I_{n_1}\oplus 0_{n_2+n_3}),0,\mathcal{R}\bigr)\,\mathcal{C}^{\mathrm{T}} = 0. 
\end{equation*}
Consequently,
\begin{equation*}\begin{aligned}
&\mathfrak{R}_2\bigl(0;\check{\mathcal{A}},I,\mathcal{C},c_1(I_{n_1}\oplus 0_{n_2+n_3}),0,\mathcal{R}\bigr) \notag \\
&\quad = c_1\begin{bmatrix}
A_{11}^{n_\diamond}(A_{11}^{n_\diamond})^{\mathrm{T}} + I_{n_1} & 0 \\ 0 & 0
\end{bmatrix}.
\end{aligned}\end{equation*}
Repeating this argument, we obtain 
\begin{equation}\begin{aligned}
&\mathfrak{R}_k\bigl(0;\check{\mathcal{A}},I,\mathcal{C},c_1(I_{n_1}\oplus 0_{n_2+n_3}),0,\mathcal{R}\bigr)\\
&\quad = c_1\begin{bmatrix}
\sum^{k-1}_{i=0}A_{11}^{in_\diamond}(A_{11}^{in_\diamond})^{\mathrm{T}}&0\\
0&0
\end{bmatrix},\ \forall\ k\geq1.
\end{aligned}\end{equation}

Combining (25) and (26) yields
\begin{equation*}
X_{11}(k n_\diamond)\geq \mathfrak{L}_k\bigl(0;A_{11}^{n_\diamond},I,c_1I\bigr). 
\end{equation*}
Since there exists \(c_2>0\) such that \(c_1I \ge c_2 \sum_{i=0}^{n_\diamond-1} A_{11}^{i}\bigl(A_{11}^{i}\bigr)^{\mathrm{T}}\), we have
\begin{equation*}
X_{11}(k n_\diamond) \ge c_2 \,\mathfrak{L}_{k n_\diamond}\bigl(0;A_{11},I,I\bigr).
\end{equation*}
For any integer \(k\ge0\) and \(\iota\in\{0,1,\cdots,n_\diamond-1\}\), an analogous argument gives
\begin{equation}
X_{11}(k n_\diamond + \iota) \ge c_2 \,\mathfrak{L}_{k n_\diamond}\bigl(0;A_{11},I,I\bigr).
\end{equation}

From (27) we obtain, for any \(k\ge1\) and \(\iota\in\{0,1,\cdots,n_\diamond-1\}\),
\begin{equation}\begin{aligned}
X_{11}(k n_\diamond + \iota) &\ge c_2 \,\mathfrak{L}_{k n_\diamond-1}\bigl(\mathfrak{L}\bigl(0;A_{11},I,I\bigr);A_{11},I,I\bigr)\\
&=c_2 \,\mathfrak{L}_{k n_\diamond-1}\bigl(I;A_{11},I,I\bigr).
\end{aligned}\end{equation}
Clearly, there exists a constant \(0<c_3\leq 1\) such that for all \(\iota\in\{0,1,\cdots,n_\diamond-1\}\),
\begin{equation}\begin{aligned}
I\geq c_3\mathfrak{L}_{\iota+1}\bigl(0;A_{11},I,I\bigr).
\end{aligned}\end{equation}
Substituting (29) into (28) yields, for any \(k\ge1\) and \(\iota\in\{0,1,\cdots,n_\diamond-1\}\),
\begin{equation*}\begin{aligned}
X_{11}(k n_\diamond + \iota) \ge c_2c_3 \,\mathfrak{L}_{k n_\diamond+\iota}\bigl(0;A_{11},I,I\bigr).
\end{aligned}\end{equation*}
This completes the proof.
\end{proof}

\subsection{The proof of Proposition 2}
Before proving Proposition 2, we introduce the following lemma. 

\begin{lemma}
Consider the difference equation 
\begin{equation*}
X(k) = JX(k) J^{\mathrm{T}} + I,\quad X(0) = 0,
\end{equation*}
where $J\in\mathbb{R}^{n\times n}$ is a real Jordan block. Its structure is shown in (2). Let $\lambda$ be an eigenvalue of $J$. Then, for some $c>0$, 
\begin{equation*}
    \mathrm{Tr}(X(k))\geq 
    \begin{cases}
        c\rho(J)^{2k}k^{2n-2},&\text{$\lambda$ is real and $\rho(J)>1$},\\
        ck^{2n-1},&\text{$\lambda$ is real and $\rho(J)=1$},\\
        c\rho(J)^{2k}k^{n-2},&\text{$\lambda$ is complex and $\rho(J)>1$},\\
        ck^{n-1},&\text{$\lambda$ is complex and $\rho(J)=1$}.
    \end{cases}
\end{equation*}
\end{lemma}

\begin{proof}
We first consider the case where $\lambda$ is real. By the nilpotent property of $J-\lambda I$, it can be shown that 
\begin{equation*}
J^i = \sum_{j=0}^{n-1} \binom{i}{j} \lambda^{i-j} (J-\lambda I)^{j},\qquad \forall i \geq 0,
\end{equation*}
where $\binom{n}{m}\triangleq\frac{n!}{m!(n-m)!}$ denotes the binomial coefficient. Moreover, note that $\mathrm{Tr}\bigl((J-\lambda I)^i((J-\lambda I)^j)^\mathrm{T}\bigr)=0$ for $i\neq j$. Consequently, we have
\begin{equation*}\begin{aligned}
\mathrm{Tr}\bigl(J^{i} (J^i)^\mathrm{T}\bigr)&=\sum_{j=0}^{n-1} (n-j) \binom{i}{j}^{2} \lambda^{2(i-j)}.
%&=\binom{i}{n-1}^{2} \lambda^{-2(n-1)}+\sum_{j=0}^{n-2} (n-j) \binom{i}{j}^{2} \lambda^{-2j}.
\end{aligned}\end{equation*}

Observe that $\lim_{i\to\infty}\binom{i}{j}/\frac{i^j}{j!}=1$, i.e., the binomial coefficient is asymptotically equivalent to $\frac{i^j}{j!}$. Hence, it follows that
\begin{equation*}\begin{aligned}
\lim_{i\to\infty}\frac{\mathrm{Tr}\bigl(J^{i} (J^i)^\mathrm{T}\bigr)}{\lambda^{2i}i^{2(n-1)}}&=\frac{1}{((n-1)!)^2\lambda^{2(n-1)}}=\text{constant}>0.
\end{aligned}\end{equation*}
By the definition of the limit, there exists a constant $c_1>0$ such that $\mathrm{Tr}\bigl(J^{i} (J^i)^\mathrm{T}\bigr)\geq c_1\lambda^{2i}i^{2(n-1)}$ for all sufficiently large $i$. A standard continuity argument extends this inequality to all $i\ge 0$ (with a smaller $c_1$).

If $\lambda=1$, we obtain for some $c_2>0$,
\begin{equation*}\begin{aligned}
\mathrm{Tr}(X(k))&=\sum^{k-1}_{i=0}\mathrm{Tr}\bigl(J^{i} (J^i)^\mathrm{T}\bigr)\geq c_1\sum^{k-1}_{i=0}i^{2(n-1)}\\
&\geq c_1\frac{(k-1)^{2n-1}}{2n-1}\geq c_2k^{2n-1},
\end{aligned}\end{equation*}
where the last inequality uses $\lim_{k\to\infty}(k-1)^p/k^p=1$ (so that $(k-1)^{2n-1}\ge c_2' k^{2n-1}$ for some $c_2'>0$ and all $k\ge1$). If $\lambda>1$, then
\begin{equation*}\begin{aligned}
\mathrm{Tr}(X(k))&=\sum^{k-1}_{i=0}\mathrm{Tr}\bigl(J^{i} (J^i)^\mathrm{T}\bigr)\geq \mathrm{Tr}\bigl(J^{k-1} (J^{k-1})^\mathrm{T}\bigr)\\
&\geq c_1\lambda^{2(k-1)}k^{2(n-1)}.
\end{aligned}\end{equation*}

When \(\lambda\) is complex, the real Jordan block \(J\) can be transformed by an invertible linear transformation into a block-diagonal form consisting of two standard Jordan blocks with eigenvalues \(\lambda\) and \(\bar{\lambda}\) and ones on the super-diagonals \cite{horn2012matrix}. The same argument then yields the stated bounds for the complex case. This completes the proof. 
\end{proof}

Then, it is ready to prove Proposition 2. 

\begin{proof}
Without loss of generality, assume that \([L_\mathrm{e}C]_{\mathbb{U}_{1j}}=0\) for some \(j\in\{1,\dots,\check{d}_1\}\). All other cases follow analogously by applying a suitable permutation similarity transformation.

If \(m_{\mathrm{e}}=0\), we have
\begin{equation*}\begin{aligned}
\hat{P}(k,L_\mathrm{e})&=\mathfrak{L}(\hat{P}(k-1,L_\mathrm{e});A,B,Q)\\
&=\mathfrak{R}(\hat{P}(k-1,L_\mathrm{e});A,B,0,Q,0,I).
\end{aligned}\end{equation*}
Using the block structure of \(A\) and \(B\) together with Lemma~1, there exist \(c>0\) and \(N\ge 0\) such that
\begin{equation*}\begin{aligned}
\mathrm{Tr}(\hat{P}(k,L_\mathrm{e}))\geq c\,\mathrm{Tr}(\mathfrak{L}_k(0;A_\mathrm{u},I,I)),\quad \forall\, k\geq N.
\end{aligned}\end{equation*}
A standard continuity argument extends this inequality to all $k\ge 0$. The claimed growth rate then follows directly from Lemma~4.

If \(m_{\mathrm{e}}\neq0\), the same partitioning of \(A\) and \(B\) together with Lemma~1 yields, for some \(c>0\),
\begin{equation*}\begin{aligned}
\mathrm{Tr}(\hat{P}(k,L_\mathrm{e}))\geq c\mathrm{Tr}\left(\mathfrak{L}_k(0;\underbrace{\begin{bmatrix}
\Lambda_1   &I          &0\\
0          &\ddots &I     \\
0          &0      &\Lambda_1
\end{bmatrix}}_{j\text{ copies of }\Lambda_1},I,I)\right).
\end{aligned}\end{equation*}
Applying Lemma~4 to the matrix inside the trace gives the desired lower bound. This completes the proof.
\end{proof}

\subsection{The proof of Proposition 3}
\begin{proof}
Define the conditional error covariance matrix 
\begin{equation*}\begin{aligned}
&\mathrel{\phantom{\triangleq}}\hat{P}(L_\mathrm{e},\gamma,k+1)\\
&\triangleq\mathrm{E}\big[(x(k+1)-\mathrm{E}[x(k+1)|\mathbb{Z}(L_\mathrm{e},\gamma,k)])\\
&\mathrel{\phantom{\triangleq}}\times(x(k+1)-\mathrm{E}[x(k+1)|\mathbb{Z}(L_\mathrm{e},\gamma,k)])^\mathrm{T}|\mathbb{Z}(L_\mathrm{e},\gamma,k)\big].
\end{aligned}\end{equation*}
Given the information set \(\mathbb{Z}(L_\mathrm{e},\gamma,k)\), Kalman filtering theory gives
\begin{equation*}\begin{aligned}
&\hat{P}(L_\mathrm{e},\gamma,k+1)=\\
&(1-\gamma(k))\mathfrak{K}(\hat{P}(L_\mathrm{e},\gamma,k),A,B,L_\mathrm{u}C,Q,SL_\mathrm{u}^{\mathrm{T}},L_\mathrm{u}RL_\mathrm{u}^\mathrm{T})\\
&+\gamma(k)\mathfrak{K}(\hat{P}(L_\mathrm{e},\gamma,k),A,B,L_\mathrm{e}C,Q,SL_\mathrm{e}^\mathrm{T},L_\mathrm{e}RL_\mathrm{e}^\mathrm{T}),
\end{aligned}\end{equation*}
where $\hat{P}(L_\mathrm{e},\gamma,0)=P(0)$. 

Consider an intermediate difference equation
\begin{equation*}\begin{aligned}
&X(k+1)=\gamma(k)\mathfrak{K}(X(k),A,B,L_\mathrm{e}C,Q,SL_\mathrm{e}^{\mathrm{T}},L_\mathrm{e}RL_\mathrm{e}^\mathrm{T})
\end{aligned}\end{equation*}
with initial condition $X(0)=P(0)$. Then, we will prove that $X(k)\leq \hat{P}(L_\mathrm{e},\gamma,k)$ for $k\geq0$ by an induction. It is trivial that $X(0)\leq \hat{P}(L_\mathrm{e},\gamma,0)$. Suppose that $X_{k}\leq \hat{P}(L_\mathrm{e},\gamma,k)$. Then, it follows from the monotonicity of the Riccati equation that
\begin{equation*}\begin{aligned}
X(k+1)&\leq\gamma(k)\mathfrak{K}(\hat{P}(L_\mathrm{e},\gamma,k),A,B,L_\mathrm{e}C,Q,SL_\mathrm{e}^{\mathrm{T}},L_\mathrm{e}RL_\mathrm{e}^\mathrm{T})\\
&\leq \hat{P}(L_\mathrm{e},\gamma,k+1).
\end{aligned}\end{equation*}
This also implies that 
\begin{equation*}\begin{aligned}
&\mathrel{\phantom{=}}\mathrm{E}[\|x(k+1)-\mathrm{E}[x(k+1)|\mathbb{Z}(L_\mathrm{e},\gamma,k)]\|^2]\\
&=\mathrm{Tr}(\mathrm{E}[\hat{P}(L_\mathrm{e},\gamma,k+1)])\\
&\geq \mathrm{Tr}(\mathrm{E}[X(k+1)]).
\end{aligned}\end{equation*}

Then, utilizing the law of total probability yields 
\begin{equation}\begin{aligned}
&\mathrel{\phantom{=}}\mathrm{E}[X(k+1)]\\
&=\sum_{\Omega\in\{0,1\}^{k+1}}\mathrm{Pr}\big((\gamma(0),\cdots,\gamma(k))=\Omega\big)\\
&\ \ \ \ \ \ \ \ \ \ \ \ \ \times\mathrm{E}\big[X_{k+1}|(\gamma(0),\cdots,\gamma(k))=\Omega\big]\\
&\geq \gamma^{k+1}\mathrm{E}\big[X(k+1)|(\gamma(0),\cdots,\gamma(k))=(1,\cdots,1)\big].
\end{aligned}\end{equation}
By (30) and Proposition~2, there exists \(c>0\) such that
\begin{equation*}\begin{aligned}
\mathrm{Tr}(\mathrm{E}[X(k)])\geq \begin{cases}
        c(\gamma\rho(J_i)^2)^kk^{2j-2},\ &\text{if $\rho(J_i)>1$,}\\
        c\gamma^kk^{2j-1},\ &\text{if $\rho(J_i)=1$.}
    \end{cases}
\end{aligned}\end{equation*}
Consequently, \(\lim_{k\to\infty}\mathrm{Tr}(\mathrm{E}[X(k)])=\infty\) whenever \(\gamma\rho(J_i)^2\ge 1\). This completes the proof.
\end{proof}

\subsection{The proof of Lemma 2}
\begin{proof}
Define $a_k = x^{\mathrm{T}} X^k (X^k)^{\mathrm{T}} x$. Suppose, for contradiction, that $\lim_{k\to\infty}\sum_{i=0}^k a_i < \infty$ for some $x \neq 0$. Since the sequence of partial sums is monotonically non-decreasing and bounded, it follows that $\lim_{k\to\infty} a_k = 0$. 

Denote the Jordan canonical form of $X$ as 
\begin{equation*}\begin{aligned}
T_x^{-1}XT_x=J_x=J_{x,1}\oplus J_{x,2}\oplus\cdots\oplus J_{x,r_x}\in\mathbb{C}^{n\times n},
\end{aligned}\end{equation*}
where $J_{x,i}\in\mathbb{C}^{d_{x,i}\times d_{x,i}}$ is the Jordan block associated with eigenvalue $\lambda_{x,i}$. Define $\check{x} = T^{\mathrm{H}}x$ with partitioned components $\check{x} = [\check{x}_1;\cdots; \check{x}_{r_x}]$ corresponding to the Jordan blocks. Since $T$ is invertible, one has $\check{x}\neq 0$, i.e., $\check{x}_i\neq0$ for some $i\in\{1,2,\cdots,r_x\}$. Then, if we can prove that $\check{x}_i\neq0$ for some $i\in\{1,2,\cdots,r_x\}$ implies that $\lim_{k\to\infty}a_k\neq0$, then the lemma can be proved. 

Suppose $\check{x}_i\neq0$ for some $i\in\{1,2,\cdots,r_x\}$. Note that the power of the Jordan block $J_{x,i}$ is given by 
\begin{equation*}
J_{x,i}^k = 
\begin{bmatrix}
\lambda_{x,i}^k & \dbinom{k}{1}\lambda_{x,i}^{k-1} & \cdots & \dbinom{k}{d_{x,i}-1}\lambda_{x,i}^{k-d_{x,i}+1} \\
0 & \lambda_{x,i}^k & \cdots & \dbinom{k}{d_{x,i}-2}\lambda_{x,i}^{k-d_{x,i}+2} \\
\vdots & \vdots & \ddots & \vdots \\
0 & 0 & \cdots & \lambda_{x,i}^k
\end{bmatrix}.
\end{equation*}

For any $0\leq j\leq d_{x,i}-1$, one can derive that 
\begin{equation*}\begin{aligned}
\lim_{k\to\infty}\frac{\dbinom{k}{j}\lambda_{x,i}^{k-j}}{\lambda_{x,i}^{k-d_{x,i}}k^j}
%&=\lim_{k\to\infty}\frac{\dbinom{k}{j}\lambda_{x,i}^{d_{x,i}-j-1}}{k^j}\\
%&=\lim_{k\to\infty}\frac{\frac{k!}{j!(k-j)!}\lambda^{r_1-j}}{k^j}\\
%&=\lim_{k\to\infty}\frac{\frac{k!}{(k-j)!}\lambda^{r_1-j}}{k^jj!}\\
%&=\lim_{k\to\infty}\frac{k(k-1)\cdots(k-j+1)\lambda_{x,i}^{d_{x,i}-j}}{k^jj!}\\
%&=\lim_{k\to\infty}\prod^{j-1}_{i=1}(1-\frac{i}{k})\frac{\lambda_{x,1}^{d_{x,1}-j}}{j!}\\
&=\frac{\lambda_{x,i}^{d_{x,i}-j}}{j!}.
\end{aligned}\end{equation*}
%Thus, one has 
%\begin{equation*}\begin{aligned}
%\dbinom{k}{j}\lambda_{x,1}^{k-j}\sim\frac{\lambda_{x,1}^{d_{x,1}-j}}{j!}\lambda_{x,1}^{k-d_{x,1}+1}k^j,
%\end{aligned}\end{equation*}
%where the symbol $\sim$ means that the two quantities are asymptotic, that is, their ratio tends to $1$ as $k$ tends to infinity. 
In this case, for different values of $j$, the term 
\[\dbinom{k}{j}\lambda_{x,i}^{k-j}\]
exhibits distinct asymptotic growth rates as $k\to \infty$. This implies that $\lim_{k\to\infty}(J_{x,i}^k)^{\mathrm{H}}\check{x}_i\neq0$. Meanwhile, note that 
\begin{equation*}\begin{aligned}
a_k\geq \lambda_{\min}(T_x^{-1}T_x^{-\mathrm{H}})\sum^{r_x}_{j=1}\check{x}_j^{\mathrm{H}}J_{x,j}^k(J_{x,j}^k)^{\mathrm{H}}\check{x}_j.
\end{aligned}\end{equation*}
Consequently, one has 
\begin{equation*}\begin{aligned}
\lim_{k\to\infty}a_k\geq\lambda_{\min}(T_x^{-1}T_x^{-\mathrm{H}})\lim_{k\to\infty}\check{x}_i^{\mathrm{H}}J_{x,i}^k(J_{x,i}^k)^{\mathrm{H}}\check{x}_i>0.
\end{aligned}\end{equation*}
The proof is completed. 
\end{proof}

\subsection{The proof of Proposition 4} 
\begin{proof}
The second statement is sufficient for the first one. It remains to show that the first statement implies the second one. 
One can derive that for some $c>0$, 
\begin{equation}\begin{aligned}
&\mathrel{\phantom{=}}\mathfrak{L}_{kn_\mathrm{c}}(0;A,B,Q)\\
&=\sum^{k}_{i=1}(A^{n_\mathrm{c}})^{i-1}\sum^{n_\mathrm{c}-1}_{j=0}A^{j}BQB^{\mathrm{T}}(A^{j})^{\mathrm{T}}((A^{n_\mathrm{c}})^{i-1})^{\mathrm{T}}\\
&\geq\sum^{k}_{i=1}(A^{n_\mathrm{c}})^{i-1}(c I_{n_\mathrm{u}}\oplus 0_{n_\mathrm{s}+n_{\bar{\mathrm{c}}}})((A^{n_\mathrm{c}})^{i-1})^{\mathrm{T}}\\
&=\begin{bmatrix}
c\sum^{k}_{i=1}(A_\mathrm{u}^{n_\mathrm{c}})^{i-1}((A_\mathrm{u}^{n_\mathrm{c}})^{i-1})^{\mathrm{T}}   &0\\
0     &0\\
\end{bmatrix}.
\end{aligned}\end{equation}

According to (31) and Lemma 2, one knows that the first $n_\mathrm{u}$ components of $p$ will be $0$ if $\lim_{k\rightarrow\infty}p^{\mathrm{T}}\mathfrak{L}_k(0;A,B,Q)p<\infty$. In this case, it can be derived that 
\begin{equation}\begin{aligned}
p^{\mathrm{T}}A^i=p^{\mathrm{T}}\begin{bmatrix}
0 &0 &0\\
0 &A_s &A_{\prime2}\\
0 &0  &A_{\bar{\mathrm{c}}}\\
\end{bmatrix}^i,\ \forall\ i\geq0.
\end{aligned}\end{equation}
Note that \[\rho\left(\begin{bmatrix}
A_s &A_{\prime2}\\
0  &A_{\bar{\mathrm{c}}}
\end{bmatrix}\right)<1.\] Thus, for any $X\geq0$, one can conclude from (32) and Lyapunov theory that 
\begin{equation}\begin{aligned}
&\mathrel{\phantom{=}}\lim_{k\rightarrow\infty}p^{\mathrm{T}}\mathfrak{L}_k(X;A,B,Q)p\\
&=\lim_{k\rightarrow\infty}p^{\mathrm{T}}\mathfrak{L}_k(X;\begin{bmatrix}
0 &0 &0\\
0 &A_s &A_{\prime2}\\
0 &0  &A_{\bar{\mathrm{c}}}\\
\end{bmatrix},B,Q)p<\infty.
%&=\lim_{k\rightarrow\infty}p^{\mathrm{T}}T^{-1}\bar{A}^kX(\bar{A}^k)^{\mathrm{T}}T^{-\mathrm{T}}p\\
%&\quad+p^{\mathrm{T}}T^{-1}\sum^{k}_{i=1}\bar{A}^{i-1}\bar{B}Q\bar{B}^{\mathrm{T}}(\bar{A}^{i-1})^{\mathrm{T}}T^{-\mathrm{T}}p\\
%&=\lim_{k\rightarrow\infty}p^{\mathrm{T}}T^{-1}\bar{A}_s^kX(\bar{A}_s^k)^{\mathrm{T}}T^{-\mathrm{T}}p\\
%&\quad+p^{\mathrm{T}}T^{-1}\sum^{k}_{i=1}\bar{A}_s^{i-1}\bar{B}Q\bar{B}^{\mathrm{T}}(\bar{A}_s^{i-1})^{\mathrm{T}}T^{-\mathrm{T}}p<\infty.
\end{aligned}\end{equation}

It follows from (31) and Lemma 2 that 
\begin{equation}\begin{aligned}
\mathbb{S}\subseteq\mathrm{Span}(e_{n,n_\mathrm{u}+1},e_{n,n_\mathrm{u}+2},\cdots,e_{n,n}).
\end{aligned}\end{equation}
Conversely, since (33) holds for every \(p\in\mathrm{Span}(e_{n,n_\mathrm{u}+1},e_{n,n_\mathrm{u}+2},\cdots,e_{n,n})\), we have 
\begin{equation}\begin{aligned}
\mathrm{Span}(e_{n,n_\mathrm{u}+1},e_{n,n_\mathrm{u}+2},\cdots,e_{n,n})\subseteq\mathbb{S}.
\end{aligned}\end{equation}
The proof is completed. 
\end{proof}

\subsection{The proof of Proposition 5}
\begin{proof}
It can be derived that 
\begin{equation}\begin{aligned}
\mathrm{E}[(x_k-\mathrm{E}[x_k])(x_k-\mathrm{E}[x_k])^{\mathrm{T}}]=\mathfrak{L}_k(P(0);A,B,Q).\nonumber
\end{aligned}\end{equation}

If $\varphi\in\mathbb{S}$, it follows from the optimality of Kalman filtering \cite{anderson2005optimal} that for any $m_\mathrm{c}$, $m_\mathrm{e}$, and $L_\mathrm{u}$, the inequality 
\begin{equation*}\begin{aligned}
\mathrm{E}[\|\varphi^{\mathrm{T}}\tilde{x}(L_\mathrm{e},k)\|^2]&\leq\mathrm{E}[\|\varphi^{\mathrm{T}}(x_k-\mathrm{E}[x_k])\|^2]\\
&=\varphi^{\mathrm{T}}\mathfrak{L}_k(P(0);A,B,Q)\varphi
\end{aligned}\end{equation*}
holds. Thus, one can conclude from Proposition 4 that 
\begin{equation*}\begin{aligned}
\lim_{k\rightarrow\infty}\mathrm{E}[\|\varphi^{\mathrm{T}}\tilde{x}(L_\mathrm{e},k)\|^2]<\infty
\end{aligned}\end{equation*}
for any $m_\mathrm{c}$, $m_\mathrm{e}$, and $L_\mathrm{u}$. 

If $\varphi\notin\mathbb{S}$, let $m_\mathrm{e}=0$. In such case, one has $\mathbb{Z}(L,k)=\emptyset$. Thus, it holds that 
\begin{equation*}\begin{aligned}
\lim_{k\rightarrow\infty}\mathrm{E}[\|\varphi^{\mathrm{T}}\tilde{x}(L_\mathrm{e},k)\|^2]=\lim_{k\rightarrow\infty}\varphi^{\mathrm{T}}\mathfrak{L}_k(P(0);A,B,Q)\varphi=\infty.
\end{aligned}\end{equation*}
The proof is completed. 
\end{proof}

\subsection{The proof of Lemma 3}
\begin{proof}
It follows from the orthogonal complement principle that $\mathrm{R}(\mathfrak{O}_n(C,A)^\mathrm{T})=\mathrm{N}(\mathfrak{O}_n(C, A))^\perp$, where $\perp$ denotes the orthogonal complement. Thus, $ \varphi\notin\mathrm{R}(\mathfrak{O}_n(C, A)^\mathrm{T})$ if and only if there exists $\beta \in \mathrm{N}(\mathfrak{O}_n(C, A))$ such that $ \varphi^\mathrm{T}\beta\neq0$. Meanwhile, the first condition is equivalent to $\beta \in \mathrm{N}(\mathfrak{O}_n(C, A))$. The proof is completed. 
\end{proof}

\subsection{The proof of Corollary 2}
\begin{proof}
If the parameters satisfy (16), it follows from Theorem 1 that $\hat{x}(L_\mathrm{u},k)=\hat{x}(I,k)$ and $\hat{P}(L_\mathrm{u},k)=\hat{P}(I,k)$ for all $k\geq1$. 

When $\vartheta_{ij}\geq\mathrm{rank}(D)$, the selection (16) gives $m_\mathrm{e}=0$. Since $\varphi_{\mathrm{u},i,j}\neq 0$ for some $i\in\{1,\cdots,r_\mathrm{u}\}$ and $j\in\{1,\cdots,\check{d}_i\}$, one has $\varphi_\mathrm{u}\neq0$. Under such cases, one can derive from Proposition 4 that 
\begin{equation*}\begin{aligned}
\lim_{k\to\infty}\mathrm{E}[\|\varphi^\mathrm{T}\tilde{x}(L_\mathrm{e},k)\|^2]=\lim_{k\to\infty}\varphi^\mathrm{T}\mathfrak{L}_k(P_0,A,B,Q)\varphi=\infty.
\end{aligned}\end{equation*}

%Without loss of generality, it is assumed that $[L_\mathrm{e}C]_{\mathbb{U}_1(j)}=0$. The other case can be proved similarly by a permutation similarity transformation. 
When $\vartheta_{ij}<\mathrm{rank}(D)$, similar to the proof of Corollary 1, one can verify that the selection (16) gives $[L_\mathrm{e}C]_{\mathbb{U}_{ij}}=0$. This also implies $[L_\mathrm{e}C_\mathrm{u}]_{\mathbb{U}_{ij}}=0$. Then, recall the partition structure (3), one has 
\begin{equation}\begin{aligned}
[L_\mathrm{e}C_\mathrm{u}A_\mathrm{u}^\ell]_{\mathbb{U}_{ij}}=0,\ \forall\ \ell=1,2,\cdots.
\end{aligned}\end{equation}
Choose the vector $\beta=[0;\cdots;0;\varphi_{\mathrm{u},i,j};0;\cdots;0]\in\mathbb{R}^{n_\mathrm{u}}$, where the sub-vector $\varphi_{\mathrm{u},i,j}$ is placed at the same position as in $\varphi_\mathrm{u}$. In this case, it can be verified from (36) that 
\begin{equation*}\begin{aligned}
&\varphi_\mathrm{u}^{\mathrm{T}}\beta=\|\varphi_{\mathrm{u},i,j}\|^2\neq 0,\\ &L_\mathrm{e}C_\mathrm{u}A_\mathrm{u}^\ell\beta=0,\ \forall\ \ell=1,2,\cdots.
\end{aligned}\end{equation*}
Therefore, by Lemma 3, we obtain $\varphi_\mathrm{u} \notin \operatorname{R}(\mathfrak{O}_{n_\mathrm{u}}(L_\mathrm{e}C_\mathrm{u},A_\mathrm{u})^{\mathrm{T}})$. Applying Theorem 3 finally yields
\[
\lim_{k\to\infty}\operatorname{E}[\|\varphi^\mathrm{T}\tilde{x}(L_\mathrm{e},k)\|^2] = \infty.
\]
The proof is completed. 
\end{proof}

\section*{Reference}
\bibliographystyle{ieeetr}
\bibliography{refs}

@article{AN2022110087,
title = {Enhancement of opacity for distributed state estimation in cyber–physical systems},
journal = {Automatica},
volume = {136},
pages = {110087},
year = {2022},
issn = {0005-1098},
doi = {https://doi.org/10.1016/j.automatica.2021.110087},
url = {https://www.sciencedirect.com/science/article/pii/S0005109821006166},
author = {Liwei An and Guang-Hong Yang}
}

@INPROCEEDINGS{Du8899554,
  author={Du, Lishuang and Zhang, Ya and Chen, Yangyang and Sun, Changyin},
  booktitle={2019 IEEE 15th International Conference on Control and Automation (ICCA)}, 
  title={A Probabilistic Scheme for Secure Estimation of Sensor Networks in the Presence of Packet Losses and Eavesdroppers}, 
  year={2019},
  volume={},
  number={},
  pages={190-195},
  doi={10.1109/ICCA.2019.8899554}}

@ARTICLE{Zhao9601251,
  author={Zhao, Bingya and Zhang, Ya and Ding, Zhengtao},
  journal={IEEE Transactions on Control of Network Systems}, 
  title={Probabilistic Transmission Scheme for Distributed Filtering Over Randomly Lossy Sensor Networks in the Presence of Eavesdropper}, 
  year={2022},
  volume={9},
  number={2},
  pages={800-810},
  doi={10.1109/TCNS.2021.3124887}}

@article{TSIAMIS20178385,
title = {State Estimation with Secrecy against Eavesdroppers},
journal = {IFAC-PapersOnLine},
volume = {50},
number = {1},
pages = {8385-8392},
year = {2017},
note = {20th IFAC World Congress},
issn = {2405-8963},
doi = {https://doi.org/10.1016/j.ifacol.2017.08.1563},
url = {https://www.sciencedirect.com/science/article/pii/S2405896317321559},
author = {Anastasios Tsiamis and Konstantinos Gatsis and George J. Pappas}
}

@ARTICLE{Leong8543618,
  author={Leong, Alex S. and Quevedo, Daniel E. and Dolz, Daniel and Dey, Subhrakanti},
  journal={IEEE Transactions on Automatic Control}, 
  title={Transmission Scheduling for Remote State Estimation Over Packet Dropping Links in the Presence of an Eavesdropper}, 
  year={2019},
  volume={64},
  number={9},
  pages={3732-3739},
  doi={10.1109/TAC.2018.2883246}}

@book{chen1984linear,
  title={Linear {S}ystem {T}heory and {D}esign},
  author={Chen, Chi-Tsong},
  year={1984},
  publisher={Saunders college publishing}
}

@book{kailath2000linear,
  title={Linear {E}stimation},
  author={Kailath, Thomas and Sayed, Ali H and Hassibi, Babak},
  year={2000},
  publisher={Prentice Hall}
}

@article{Anderson0319002,
author = {Anderson, B. D. O. and Moore, J. B.},
title = {Detectability and Stabilizability of Time-Varying Discrete-Time Linear Systems},
journal = {SIAM Journal on Control and Optimization},
volume = {19},
number = {1},
pages = {20-32},
year = {1981},
doi = {10.1137/0319002},
URL = {https://doi.org/10.1137/0319002},
eprint = {https://doi.org/10.1137/0319002}
}

@ARTICLE{Shang9882330,
  author={Shang, Jun and Chen, Tongwen},
  journal={IEEE Transactions on Automatic Control}, 
  title={Linear Encryption Against Eavesdropping on Remote State Estimation}, 
  year={2023},
  volume={68},
  number={7},
  pages={4413-4419},
  doi={10.1109/TAC.2022.3205548}}

@book{anderson2005optimal,
  title={Optimal {F}iltering},
  author={Anderson, Brian DO and Moore, John B},
  year={2005},
  publisher={Courier Corporation}
}

@ARTICLE{Yan10648957,
  author={Yan, Xinhao and Zhou, Guanzhong and Quevedo, Daniel E. and Murguia, Carlos and Chen, Bo and Huang, Hailong},
  journal={IEEE Transactions on Automation Science and Engineering}, 
  title={Privacy-Preserving State Estimation in the Presence of Eavesdroppers: A Survey}, 
  year={2024},
  volume={},
  number={},
  pages={1-18},
  doi={10.1109/TASE.2024.3440042}}

@ARTICLE{Le6606817,
  author={Le Ny, Jerome and Pappas, George J.},
  journal={IEEE Transactions on Automatic Control}, 
  title={Differentially Private Filtering}, 
  year={2014},
  volume={59},
  number={2},
  pages={341-354},
  doi={10.1109/TAC.2013.2283096}}

@ARTICLE{Degue9993779,
  author={Degue, Kwassi Holali and Le Ny, Jerome},
  journal={IEEE Transactions on Automatic Control}, 
  title={Differentially Private {K}alman Filtering With Signal Aggregation}, 
  year={2023},
  volume={68},
  number={10},
  pages={6240-6246},
  doi={10.1109/TAC.2022.3230735}}

@INPROCEEDINGS{Leong8550317,
  author={Leong, Alex S. and Redder, Adrian and Quevedo, Daniel E. and Dey, Subhrakanti},
  booktitle={2018 European Control Conference (ECC)}, 
  title={On the Use of Artificial Noise for Secure State Estimation in the Presence of Eavesdroppers}, 
  year={2018},
  volume={},
  number={},
  pages={325-330},
  doi={10.23919/ECC.2018.8550317}}

@ARTICLE{Tsiamis8758381,
  author={Tsiamis, Anastasios and Gatsis, Konstantinos and Pappas, George J.},
  journal={IEEE Transactions on Automatic Control}, 
  title={State-Secrecy Codes for Networked Linear Systems}, 
  year={2020},
  volume={65},
  number={5},
  pages={2001-2015},
  doi={10.1109/TAC.2019.2927459}}

@ARTICLE{Kennedy10491308,
  author={Kennedy, Justin M. and Ford, Jason J. and Quevedo, Daniel E. and Dressler, Falko},
  journal={IEEE Transactions on Automatic Control}, 
  title={Innovation-Based Remote State Estimation Secrecy With No Acknowledgments}, 
  year={2024},
  volume={69},
  number={11},
  pages={7433-7448},
  doi={10.1109/TAC.2024.3385315}}

@ARTICLE{Hassan8854247,
  author={Hassan, Muneeb Ul and Rehmani, Mubashir Husain and Chen, Jinjun},
  journal={IEEE Communications Surveys \& Tutorials}, 
  title={Differential Privacy Techniques for Cyber Physical Systems: A Survey}, 
  year={2020},
  volume={22},
  number={1},
  pages={746-789},
  doi={10.1109/COMST.2019.2944748}}

@ARTICLE{Lücke9762536,
  author={Lücke, Marvin and Lu, Jingyi and Quevedo, Daniel E.},
  journal={IEEE Transactions on Automatic Control}, 
  title={Coding for Secrecy in Remote State Estimation With an Adversary}, 
  year={2022},
  volume={67},
  number={9},
  pages={4955-4962},
  doi={10.1109/TAC.2022.3169839}}

@ARTICLE{Marelli10684094,
  author={Marelli, Damián and Sui, Tianju and Fu, Minyue and Cai, Qianqian},
  journal={IEEE Transactions on Automatic Control}, 
  title={Secrecy Codes for State Estimation of General Linear Systems}, 
  year={2025},
  volume={70},
  number={2},
  pages={1161-1168},
  doi={10.1109/TAC.2024.3463728}}

@article{CRIMSON2025111932,
title = {Remote state estimation with privacy against active eavesdroppers},
journal = {Automatica},
volume = {171},
pages = {111932},
year = {2025},
issn = {0005-1098},
doi = {https://doi.org/10.1016/j.automatica.2024.111932},
url = {https://www.sciencedirect.com/science/article/pii/S0005109824004266},
author = {Matthew J. Crimson and Justin M. Kennedy and Daniel E. Quevedo}
}

@article{YANG2020109116,
title = {An encoding mechanism for secrecy of remote state estimation},
journal = {Automatica},
volume = {120},
pages = {109116},
year = {2020},
issn = {0005-1098},
doi = {https://doi.org/10.1016/j.automatica.2020.109116},
url = {https://www.sciencedirect.com/science/article/pii/S0005109820303149},
author = {Wen Yang and Dengke Li and Heng Zhang and Yang Tang and Wei Xing Zheng}
}

@ARTICLE{Shim7172449,
  author={Shim, Kyung-Ah},
  journal={IEEE Communications Surveys  Tutorials}, 
  title={A Survey of Public-Key Cryptographic Primitives in Wireless Sensor Networks}, 
  year={2016},
  volume={18},
  number={1},
  pages={577-601},
  doi={10.1109/COMST.2015.2459691}}

@article{MISHRA2024101037,
title = {A survey on security and cryptographic perspective of Industrial-Internet-of-Things},
journal = {Internet of Things},
volume = {25},
pages = {101037},
year = {2024},
issn = {2542-6605},
doi = {https://doi.org/10.1016/j.iot.2023.101037},
url = {https://www.sciencedirect.com/science/article/pii/S2542660523003608},
author = {Nimish Mishra and SK {Hafizul Islam} and Sherali Zeadally}
}

@article{HUANG2021109537,
title = {Encryption scheduling for remote state estimation under an operation constraint},
journal = {Automatica},
volume = {127},
pages = {109537},
year = {2021},
issn = {0005-1098},
doi = {https://doi.org/10.1016/j.automatica.2021.109537},
url = {https://www.sciencedirect.com/science/article/pii/S0005109821000571},
author = {Lingying Huang and Kemi Ding and Alex S. Leong and Daniel E. Quevedo and Ling Shi},
}

@article{WANG2022110145,
title = {Transmission scheduling for privacy-optimal encryption against eavesdropping attacks on remote state estimation},
journal = {Automatica},
volume = {137},
pages = {110145},
year = {2022},
issn = {0005-1098},
doi = {https://doi.org/10.1016/j.automatica.2021.110145},
url = {https://www.sciencedirect.com/science/article/pii/S0005109821006749},
author = {Le Wang and Xianghui Cao and Heng Zhang and Changyin Sun and Wei Xing Zheng},
}

@ARTICLE{Tao10621055,
  author={Tao, Fei and Ye, Dan},
  journal={IEEE Transactions on Industrial Informatics}, 
  title={Optimal Encryption Scheduling Policy Against Eavesdropping Attacks in Cyber-Physical Systems}, 
  year={2024},
  volume={20},
  number={11},
  pages={13147-13157},
  doi={10.1109/TII.2024.3431096}}

@ARTICLE{Zou10782997,
  author={Zou, Lei and Wang, Zidong and Shen, Bo and Dong, Hongli},
  journal={IEEE Transactions on Automatic Control}, 
  title={Secure Recursive State Estimation of Networked Systems Against Eavesdropping: A Partial-Encryption-Decryption Method}, 
  year={2024},
  volume={},
  number={},
  pages={1-14},
  doi={10.1109/TAC.2024.3512413}}

@book{horn2012matrix,
  title={Matrix {A}nalysis},
  author={Horn, Roger A and Johnson, Charles R},
  year={2012},
  publisher={Cambridge university press}}

@ARTICLE{Wang9678137,
  author={Wang, Yamin and Lam, James and Lin, Hong},
  journal={IEEE Transactions on Cybernetics}, 
  title={Consensus of Linear Multivariable Discrete-Time Multiagent Systems: Differential Privacy Perspective}, 
  year={2022},
  volume={52},
  number={12},
  pages={13915-13926},
  doi={10.1109/TCYB.2021.3135933}}
\end{document}